\documentclass{article}
\usepackage[utf8]{inputenc}
\usepackage{amsmath,amssymb,amsfonts,stmaryrd}
\usepackage{physics}
\usepackage{dsfont}
\usepackage{graphicx}
\usepackage{xcolor}
\usepackage{graphicx}
\usepackage{cancel}
\usepackage[numbers,sort&compress,square]{natbib}
\usepackage{hyperref}
\hypersetup{
    colorlinks=true,
    linkcolor=black,
    filecolor=red,      
    urlcolor=blue,
    citecolor=red
    }
\usepackage{bm} %allows for bold math type
\usepackage{subfigure} %allows for 2x2 figures
\usepackage{environ}
\usepackage{url}
\usepackage[margin=0.75in]{geometry}
\usepackage[mathscr]{euscript}
\usepackage[scr=boondox]{mathalpha}
\usepackage{amsthm}

\newcommand{\Z}{{\mathbb Z}}

\NewEnviron{eqs}{%
\begin{align}\begin{split}
    \BODY
\end{split}\end{align}
}
\newtheorem{theorem}{Theorem}[section]

\title{
Automorphism in Gauge Theories: Higher Symmetries and Transversal Non-Clifford Logical Gates
}
\author{Po-Shen Hsin$^{1}$, Ryohei Kobayashi$^{2}$}
\date{\today}

\begin{document}

\maketitle

\begin{center}

${}^1$ Department of Mathematics, King’s College London, Strand, London WC2R 2LS, UK \\
 ${}^2$School of Natural Sciences, Institute for Advanced Study, Princeton, NJ 08540, USA
 \end{center}

\bigskip\bigskip
\begin{abstract}
Gauge theories are important descriptions for many physical phenomena and systems in quantum computation. Automorphism of gauge group naturally gives global symmetries of gauge theories. In this work we study such symmetries in gauge theories induced by automorphisms of the gauge group,
when the gauge theories have nontrivial topological actions in different spacetime dimensions. We discover the automorphism symmetry can be extended, become a higher group symmetry, and/or become a non-invertible symmetry. We illustrate the discussion with various models in field theory and on the lattice.
In particular, we use automorphism symmetry to construct new transversal non-Clifford logical gates in topological quantum codes. In particular, we show that 2+1d $\mathbb{Z}_N$ qudit Clifford stabilizer models can implement non-Clifford transversal logical gate in the 4th level $\mathbb{Z}_N$ qudit Clifford hierarchy for $N\geq 3$, extending the generalized Bravyi-K\"onig bound proposed in \cite{Kobayashi:2025cfh} for qubits.

\medskip
\noindent
\end{abstract}

\bigskip \bigskip \bigskip
 
% \date{\today}

\bigskip

% \eject

\tableofcontents

\unitlength = .8mm

\setcounter{tocdepth}{3}

\bigskip

\section{Introduction}
\label{sec:intro}

Symmetries in finite group gauge theories are important in exploring new gapped phases enriched with symmetries, which have wide applications in quantum computation and dynamics of quantum systems, as well as mathematical structures of higher fusion categories. Gauge theories with gauge group $G$ can have topological action, or twist, described by SPT phases with $G$ symmetry \cite{Chen:2011pg}. Here, we will focus on Dijkgraaf-Witten theories, where the topological action in $D$ spacetime dimension is described by $D$-cocycle for group $G$ \cite{Dijkgraaf:1989pz}. 

There are several important sources of symmetries in finite group gauge theories. One of them is magnetic defects that generate one-form symmetry, and various gauged SPT defects including the Wilson operators. In addition, there are automorphism symmetries given by the automorphism group of the gauge group. More precisely, the faithful symmetry is the outer automorphism, given by the automorphism group quotient by the inner automorphisms, i.e. gauge transformation. Here, we will focus on the entire automorphism group, which is an extension of the outer automorphism group by the inner automorphisms. Once extended, the symmetry can be gauged when there is no topological action---the resulting gauge group is the semidirect product of the original gauge group and the automorphism. The automorphism symmetry, being a 0-form symmetry, can be expressed as condensation defects \cite{Gaiotto:2019xmp} of the Wilson lines and the magnetic operators (see e.g. \cite{Cordova:2024mqg} for a recent discussion).

When the gauge theory has a topological action, the magnetic defect is modified by gauged SPT defects and can become non-invertible \cite{Barkeshli:2022edm}. 
In addition, some of the gauged SPT defects become trivial \cite{Barkeshli:2022edm}. However, the fate of automorphism symmetry in the presence of topological action has not being systematically investigated.
This work serves to fill in the gap for automorphism symmetries. In a separate work we will extend the discussion to the quotient outer automorphisms. We will show that in the presence of topological action for the gauge group, the 0-form symmetry from the automorphism of gauge group can be modified in the following ways:
\begin{itemize}
    \item[1. ] The group automorphism remains an invertible 0-form symmetry, but the fusion rules are modified: the 0-form symmetry is extended by another 0-form symmetry. See section \ref{sec:Z24in2+1d} for an example.
    
    \item[2. ] The group automorphism remains an invertible 0-form symmetry of the theory, the fusion rules are not modified, but the 0-form symmetry participates in nontrivial higher group symmetries \cite{Cordova:2018cvg,Benini:2018reh}: the 0-form symmetry is extended by higher-form symmetries. See section \ref{subsub:maxwell} and section \ref{sec:WWhighergroup} for examples.
    
    \item[3. ] The group automorphism becomes a non-invertible 0-form symmetry.\footnote{
    See e.g. \cite{Shao:2023gho,Schafer-Nameki:2023jdn} for an introduction to non-invertible symmetries.
    } See section \ref{subsub:maxwell} and section \ref{sec:WWhighergroup} for examples.
\end{itemize}

An important application of finite group gauge theories is topological quantum codes as pioneered by Kitaev \cite{Kitaev:1997wr}. Symmetries in finite group gauge theories serve as logical gates in topological quantum codes, as recently studied in \cite{Yoshida_gate_SPT_2015, 
 Yoshida_global_symmetry_2016, Yoshida2017387, Barkeshli:2022edm,Barkeshli:2023bta,Kobayashi2024crosscap,Hsin2024_non-Abelian,barkeshli2023codimension,Hsin:2024nwc,kobayashi2023fermionic, zhu2025topological, zhu2025transversal}. 
 Invertible symmetries in quantum codes can be finite depth circuits that give rise to transversal logical gates.
 We will use automorphism symmetries to construct new transversal non-Clifford logical gates. Examples of such transversal gates from automorphism symmetries have been studied in our previous works \cite{Hsin2024_non-Abelian,Kobayashi:2025cfh}. In particular, in \cite{Kobayashi:2025cfh} we discuss the construction of transversal $R_D$ logical gates in $D$ spacetime dimensions using our method. We also conjectured in \cite{Kobayashi:2025cfh} that the transversal logical gates in Clifford hierarchy $h$ can be realized in spatial dimension $(h-1)$ for qubit systems. In this work, we will show that $\mathbb{Z}_N$ qudit allows us to further realize transversal logical gates of higher level in qudit Clifford hierarchy.

The work is organized as follows. In section~\ref{sec:automorphismgate}, we discuss automorphism symmetry in twisted gauge theories.
We show that automorphism symmetry can (1) become extended as illustrated by $\mathbb{Z}_2^4$ gauge theory in section \ref{sec:Z24in2+1d}, (2) become a higher group as illustrated by $\mathbb{Z}_2^2$ 2-form gauge theory in section \ref{sec:highergroupexample} and Maxwell theory in section \ref{subsub:maxwell}, and (3) become a non-invertible symmetry as illustrated by $\mathbb{Z}_2^3$ gauge theory in section \ref{subsec:lattice_sandwich} and Maxwell theory in section \ref{subsub:maxwell}.
In section \ref{sec:higher group gauged SPT}, we discuss its interplay with gauged SPT symmetries in untwisted gauge theories such as $\mathbb{Z}_N\times\mathbb{Z}_M$ gauge theory in section \ref{subsubsec:ZNZM in generic dim}.
In section \ref{sec:twistedZNm}, we discuss the automorphism symmetry in $\mathbb{Z}_N^m$ gauge theories. 
In section \ref{sec:noncliffordgate}, we discuss the application of automorphism symmetry for constructing transversal non-Clifford logical gates. For more applications in logical gates, see our companion paper \cite{Kobayashi:2025cfh}. In section \ref{sec:discussion}, we discuss several future directions. In section \ref{sec:discussion}, we discuss a few future directions.

\section{Automorphism in Twisted Gauge Theories and Higher Symmetry}
\label{sec:automorphismgate}

\begin{figure}[t]
    \centering
    \includegraphics[width=0.4\linewidth]{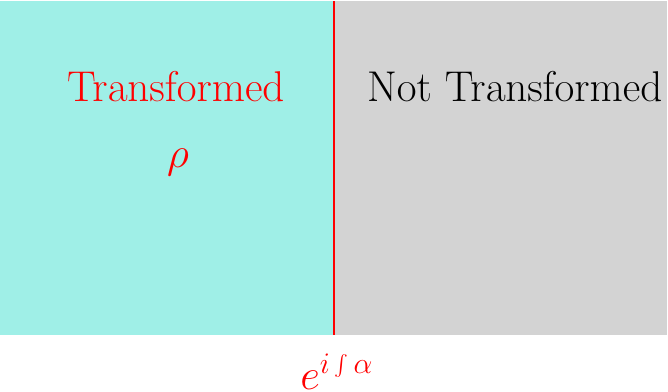}
    \caption{Automorphism symmetry $\rho:G\rightarrow G$ in twisted gauge theory with topological action $\omega$ needs to decorate with gauged SPT symmetry $e^{i\int \alpha}$ for automorphism $\rho$ that preserves the cohomology class $[\omega]$ but in general not the cocycle, $\rho^*\omega=\omega+d\alpha$. The time direction goes from the right to the left, and the interface is oriented.}
    \label{fig:automorphism}
\end{figure}

Consider $G$ gauge theory with automorphism $\rho:G\rightarrow G$. The theory has topological action given by cocycle $\omega$ with $[\omega]\in H^D(G,U(1))$ in spacetime dimension $D$.
For trivial topological action $\omega=0$, the automorphism is an invertible symmetry of the theory. For nontrivial topological action $\omega$, there can be two situations:
\begin{itemize}
    \item[1.] The automorphism is an invertible symmetry if it preserves the cohomology class 
\begin{equation}
    [\rho^*\omega]=[\omega]~.
\end{equation}
The cocycles $\omega,\rho^*\omega$ can differ by an exact cocycle:
\begin{equation}
    \rho^*\omega=\omega+d\alpha_\rho~,
\end{equation}
where $\alpha_\rho$ is a $U(1)$ valued group cochain. When we perform the automorphism transformation on half spacetime as in Fig.~\ref{fig:automorphism}, the action is shifted by $\int d\alpha$, and via Stokes' theorem it gives rise to additional gauged SPT defect $e^{i\int \alpha}$ on the domain wall that generates the automorphism symmetry. In other words, the automorphism symmetry is
\begin{equation}
    U_\rho= e^{i\int \alpha_\rho} V_\rho~,
\end{equation}
where $V_\rho$ implements the automorphism action.

We emphasize that the gauged SPT defect $e^{i\int\alpha}$ is not topological on its own; otherwise, it can be removed simply by attaching with another topological gauged SPT defect, since gauged SPT defects generate invertible symmetries. Likewise, a pure automorphism action $V_\rho$ is not a symmetry of the theory. Only when combined together they give the well-defined automorphism symmetry domain wall.

We also note that the extra SPT factor $e^{i\int \alpha}$ in the automorphism symmetry is reminiscent of the magnetic defects in twisted gauge theories acquiring SPT factor \cite{Barkeshli:2022edm}: both arise because of the twist in the theory.

\item[2.] When no such $\alpha$ exists, the domain wall that generates the automorphism symmetry needs to decorate with additional TQFT to cancel the anomaly associated with the SPT class $[\rho^*\omega-\omega]$.
Such TQFT makes the automorphism symmetry non-invertible, similar to the construction in \cite{Barkeshli:2022edm,Hsin:2022heo}. For an introduction to non-invertible symmetries in general quantum systems, see e.g. \cite{Schafer-Nameki:2023jdn,Shao:2023gho}. 
\end{itemize}

Let us begin with the first case when the automorphism symmetry is invertible. Later in section \ref{sec:non-invertible} we will present a sandwich construction for general automorphism non-invertible symmetry.

\paragraph{Example: swap automorphism symmetry with additional gauged SPT defect}

To illustrate the additional gauged SPT defect on the automorphism symmetry, let's consider $\mathbb{Z}_N\times\mathbb{Z}_N$ gauge theory in 2+1d with the mixed Chern-Simons topological action
\begin{equation}\label{eqn:ZNxZNCS}
    \omega=\frac{k}{2\pi}a_1 da_2~,
\end{equation}
where $a_1,a_2$ are gauge fields for the two $\mathbb{Z}_N$s.
Consider the swap automorphism that exchanges the two $\mathbb{Z}_N$s:
\begin{equation}
\rho:\quad (g_1,g_2)\rightarrow (g_2,g_1)\in\mathbb{Z}_N\times\mathbb{Z}_N~.
\end{equation}
The automorphism symmetry acts on the gauge fields as $a_1\leftrightarrow a_2$. Under the automorphism symmetry, the mixed Chern-Simons action transforms as
\begin{equation}
    \rho^*\omega= \frac{k}{2\pi}a_2 da_1=\frac{k}{2\pi}a_1da_2+d\alpha,\quad \alpha=\frac{k}{2\pi}a_1 a_2~.
\end{equation}
Thus the automorphism symmetry is attached with the gauged SPT defect as
\begin{equation}
    U_\rho=e^{\frac{ik}{2\pi}\int a_1 a_2}V_\rho~.
\end{equation}

\subsection{Automorphism symmetry can get extended by gauged SPT symmetries}
\label{subsec:automorphism extended by SPT}
\begin{figure}
    \centering
    \includegraphics[width=0.3\linewidth]{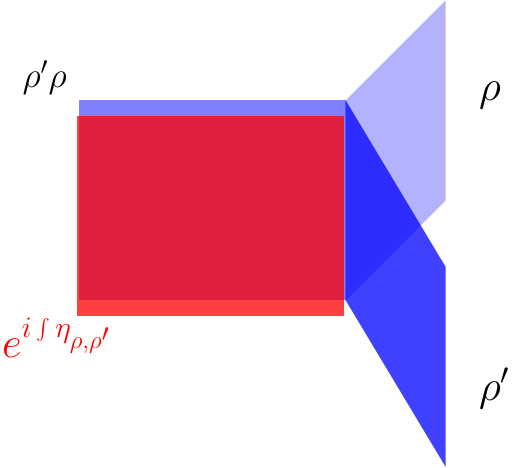}
    \caption{Automorphism symmetry can be extended by the gauged SPT symmetry $e^{i\int \eta}$ due to the decoration by gauged SPT defect $U_\rho=e^{i\int \alpha_\rho}V_\rho$.}
    \label{fig:extension}
\end{figure}

The modification on the automorphism symmetry can change its fusion rule. When the automorphism symmetries $\rho,\rho'$ fuse together, this gives
\begin{align}
    U_{\rho'}(M)\cdot U_\rho(M)&=e^{i\int_M \left(\alpha_{\rho'}+\rho'(\alpha_{\rho})\right)}  U_{\rho'}^0(M)\cdot U^0_{\rho}(M)\cr 
    &=
    e^{i\int_M \left(\alpha_{\rho'}+\rho'(\alpha_\rho)\right)}U_{\rho'\rho}^0(M)\cr 
    &=e^{i\int_M\left(\alpha_{\rho'}+\rho'(\alpha_\rho)-\alpha_{\rho'\rho}\right) }U_{\rho'\rho}(M)~,
\end{align}
where $\alpha_\rho$ is the group cochain associated with automorphism $\rho$, and $\rho'(\alpha_\rho)$ is its pullback under the automorphism action $\rho'$. We used $U^0_{\rho'}\cdot U^0_\rho=U^0_{\rho'\rho}$.
Thus the automorphism symmetry composes projectively with the projective factor (see Fig.~\ref{fig:extension})
\begin{equation}
    e^{i\eta_{\rho,\rho'}}:=e^{i\left(\alpha_{\rho'}+\rho'(\alpha_\rho)-\alpha_{\rho'\rho}\right)}~.
\end{equation}
This projective factor itself is an operator: it is a gauged SPT operator supported on $(D-1)$-dimensional submanifold in $D$ spacetime dimension. Indeed $\eta_{\rho,\rho'}$ is closed,
\begin{align}
    d\eta_{\rho,\rho'} = (\rho'^*\omega-\omega) + \rho'^*(\rho^*\omega-\omega ) - ((\rho'\rho)^*\omega-\omega) = 0~,
\end{align}
therefore the operator $e^{i\int\eta_{\rho,\rho'}}$ is topological.
In other words, the junction of 0-form automorphism symmetries emits additional 0-form gauged SPT symmetry operator, and the automorphism 0-form symmetry is extended by this 0-form gauged SPT symmetry $e^{i\int \eta}$. 

\subsubsection{Example: $\mathbb{Z}_2^4$ gauge theory in 2+1d }
\label{sec:Z24in2+1d}

Consider $\mathbb{Z}_2^4$ gauge theory with the action
\begin{equation}
\omega=\frac{\pi}{2}   \left(a_1 \cup da_4 -da_3 \cup a_2\right)~,
\end{equation}
where $a_i$ are the $\mathbb{Z}_2$ gauge fields for the four $\mathbb{Z}_2$s in the gauge group.
The following automorphism of the gauge group has order 2:
\begin{equation}
    \rho: (g_1,g_2,g_3,g_4)\rightarrow  (g_1+g_3,g_2+g_4,g_3, g_4)\ \mod 2, \quad  \rho^2=1~.
\end{equation}
The automorphism transforms the gauge fields as $a_3\rightarrow a_3+a_1,a_4\rightarrow a_4+a_2$.
The action transforms under the automorphism as
\begin{equation}
    \omega\rightarrow \omega+d\alpha,\quad \alpha=\frac{\pi}{2}a_1\cup a_2~.
\end{equation}
Thus the automorphism symmetry is
\begin{equation}
    U_\rho=i^{\int a_1\cup a_2}V_\rho~.
\end{equation}
The automorphism symmetry becomes a $\mathbb{Z}_4$ symmetry: the square of the symmetry gives a nontrivial order 2 operator:
\begin{equation}
    U_\rho(M)^2=(-1)^{\int_M a_1\cup a_2}~,
\end{equation}
which is a nontrivial gauged SPT operator that pumps the cluster state for the first two $\mathbb{Z}_2$s in $\mathbb{Z}_2^3$ on the 1+1d support $M$ of the automorphism symmetry operator. 
In other words, the $\mathbb{Z}_2$ automorphism $\rho$ of the gauge group is extended by the gauged SPT symmetry to become a $\mathbb{Z}_4$ symmetry.

\subsubsection{Example of automorphism symmetry not extended}

Whether the automorphism symmetry is extended depends on the cocycle and the particular automorphism. For example, consider another cocycle of $\mathbb{Z}_2^3$ gauge group in 2+1d:
\begin{equation}\label{eqn:a1a2a3}
    \omega=\pi a_1\cup a_2\cup a_3~,
\end{equation}
and the $\mathbb{Z}_2$ automorphism
\begin{equation}
    \rho:(g_1,g_2,g_3)\rightarrow (g_1+g_2,g_2,g_3+g_2)\ \mod 2.
\end{equation}
We remark that the theory and the automorphism symmetry has been considered in \cite{Kobayashi:2025cfh} to construct transversal T logical gate of $\Z_2^3$ gauge theory in 2+1d. The automorphism transforms the gauge fields as $a_2\rightarrow a_2+a_1+a_3$.

As discussed in \cite{Kobayashi:2025cfh}, the action transforms as
\begin{equation}
    \omega\rightarrow\omega+d\alpha,\quad \alpha=\frac{\pi}{2}a_1\cup a_3~.
\end{equation}
Thus the automorphism symmetry is
\begin{equation}
    U_\rho(M)=i^{\int_M a_1\cup a_3}V_\rho~.
\end{equation}
The automorphism squares to 1:
\begin{equation}
    U_\rho(M)^2=(-1)^{\int_M a_1\cup a_3}=1~,
\end{equation}
which is a trivial domain wall operator, since the twist (\ref{eqn:a1a2a3}) implies that on closed $M$, $\int_M a_1\cup a_3=0$.

\subsection{Automorphism symmetry can form higher group}
\label{sec:highergroupexample}
Due to the decoration with gauged SPT defect $e^{i\int \alpha}$, the automorphism symmetry can also become an extension that involves higher-form symmetry---the combined symmetry becomes a higher group. This happens when the junction of automorphism symmetries produce a gauged SPT symmetry supported on lower-dimensional submanifolds. We will give an example to illustrate such situation.

Consider $\mathbb{Z}_2\times\mathbb{Z}_2$ 2-form gauge theory in 4+1d, with the topological action
\begin{equation}
    \omega=\frac{\pi}{2}a_1\cup da_2~,
\end{equation}
where $a_1,a_2$ are the $\mathbb{Z}_2$ 2-form gauge fields. 
In the following, we rewrite the theory using the continuous notation by embedding the gauge fields inside continuous $U(1)$ 2-form gauge fields: $a_i\rightarrow a_i/\pi$ with $\oint a_i=0,\pi$ mod $2\pi$. The constraint on the holonomy can be enforced by Lagrangian multipliers:
\begin{equation}
    \frac{1}{2\pi}a_1da_2 +\frac{2}{2\pi}a_1d\tilde a_1+\frac{2}{2\pi}a_2d\tilde a_2~.
\end{equation}
Integrating out $\tilde a_i$ recovers the action for the $\mathbb{Z}_2\times\mathbb{Z}_2$ 2-form gauge fields.

Let us focus on the automorphism of $\mathbb{Z}_2\times\mathbb{Z}_2$ given by
\begin{equation}
    \rho:\quad (g_1,g_2)\rightarrow (g_1+g_2, g_2) \ \mod 2,\quad 
    \rho^2=1~.
\end{equation}
The automorphism transforms the gauge fields as $a_2\rightarrow a_2+a_1$.
Under the automorphism, the topological action transforms as (we still use the continuous notation)
\begin{equation}
    \omega\rightarrow\omega+\frac{1}{2\pi}a_1da_1=\omega+d\alpha,\quad 
\alpha=\frac{1}{4\pi}a_1a_1~.
\end{equation}
Thus the automorphism symmetry is
\begin{equation}
    U_\rho=e^{\frac{i}{4\pi}\int a_1 a_1}V_\rho~.
\end{equation}
The extra gauged SPT factor $e^{i\int\alpha}$ represents half of the minimal 1-form symmetry SPT phase with $\mathbb{Z}_2$ 1-form symmetry \cite{Kapustin:2014gua,Hsin:2018vcg,Tsui:2019ykk}.

\paragraph{Symmetry extension as 0-form symmetry}

First, let us show that as a 0-form symmetry, the automorphism symmetry $U_\rho$ is actually extended to be $\mathbb{Z}_4$ symmetry for the $\mathbb{Z}_2$ automorphism $\rho$.
Taking the square of the symmetry operator on 4-dimensional submanifold $M$ gives
\begin{equation}
    U_\rho(M)^2=e^{\frac{2i}{4\pi}\int_M a_1 a_1}~,
\end{equation}
which pumps a well-defined $\mathbb{Z}_2$ 1-form symmetry SPT phase on the submanifold $M$ \cite{Kapustin:2014gua,Hsin:2018vcg,Tsui:2019ykk}. The right hand side generates a $\mathbb{Z}_2$ 0-form symmetry, and thus the automorphism symmetry is extended to be $\mathbb{Z}_4$ 0-form symmetry.

\paragraph{Higher group symmetry: combining with 2-form symmetry}

Let us further take the 4th power of the automorphism symmetry $U_\rho$:
\begin{equation}
    U_\rho(M)^4=e^{\frac{4i}{4\pi}\int_M a_1 a_1}=(-1)^{\int_M \frac{a_1}{\pi}\frac{a_1}{\pi}}=(-1)^{\int \frac{a_1}{\pi}\cup (w_2+w_1^2)}=e^{i\int_\gamma a_1}~,
\end{equation}
where $w_i$ are the $i$th Stiefel-Whitney classes of $TM$, and we used the Wu formula on 4-manifold $M$: $x\cup x=x\cup (w_2+w_1^2)$ for $\mathbb{Z}_2$ gauge fields $x=0,1$ \cite{milnor1974characteristic}. The right hand side is equivalent to a membrane operator on $\gamma$ given by the Poincar\'e dual of $(w_2+w_1^2)$ on $M$. Since the membrane operator on the right hand side generates a higher 2-form symmetry, the automorphism 0-form symmetry is mixed with this 2-form symmetry to become a 3-group symmetry.

For example, if we take the submanifold $M=\mathbb{CP}^2$, the fusion of four automorphism 0-form symmetries produce the 2-form symmetry supported on the dual 2-cycle of the K\"ahler form of $\mathbb{CP}^2$. On the other hand, for the submanifold $M=T^4$ or $M=S^2\times S^2$, the membrane operator produced by the fusion is trivial.

We remark that the above discussion is similar to the 3-group symmetry in 4+1d loop toric code 2-form gauge theory discussed in \cite{Chen:2021xuc}.

\subsection{Automorphism non-invertible symmetry from a sandwich construction}
\label{sec:non-invertible}

The automorphism symmetry $\rho$ changes the topological action for $G$ gauge theory by
\begin{equation}
    \rho^*\omega-\omega~,
\end{equation}
which is a cocycle of group $G$. Denote the subgroup $K$ such that the cocycle is exact: under inclusion $\iota:K\rightarrow G$,
\begin{equation}
    \iota\left(\rho^*\omega-\omega\right)=d\alpha~,
\end{equation}
where $\alpha$ is a cochain of subgroup $K$.
Then we can express the automorphism symmetry as the following sandwich construction similar to \cite{Hsin:2024aqb,Hsin:2025ria}: in the middle, it is a domain wall that generates an invertible symmetry in $K$ gauge theory, and on the left and right it is the condensation interface where the $G$ gauge group is broken to $K$ subgroup (see Fig.~\ref{fig:sandwichautomorphism}).

\begin{figure}[t]
    \centering
    \includegraphics[width=0.4\linewidth]{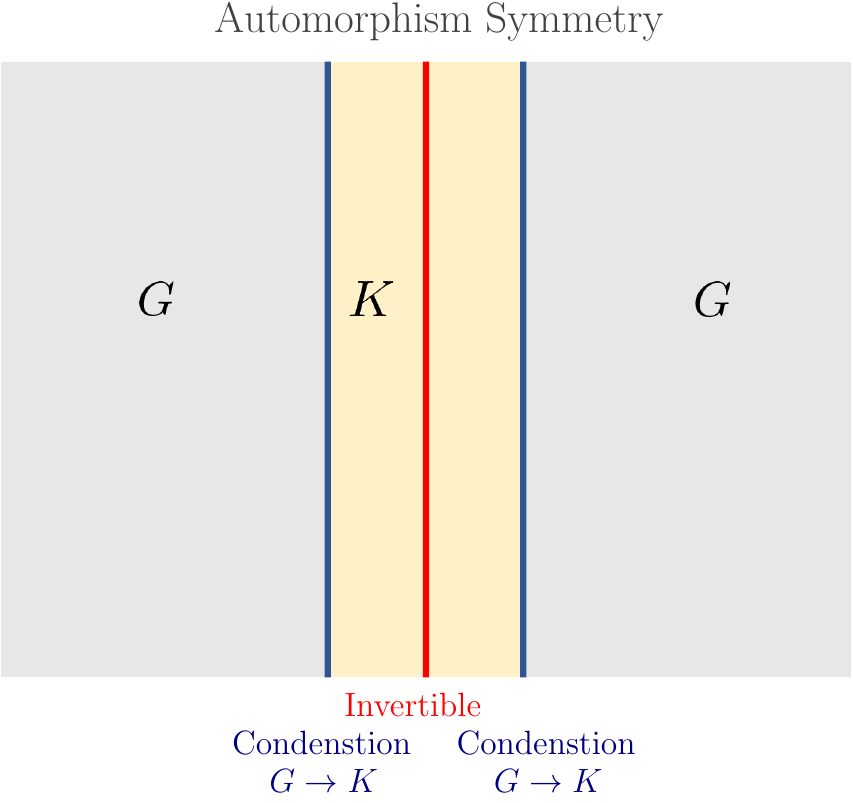}
    \caption{Sandwich construction for automorphism symmetry in twisted gauge theory with gauge group $G$. The two blue interfaces break the gauge group from $G$ to subgroup $K$, such that the automorphism symmetry becomes invertible higher-group symmetry in gauge group $K$. As a whole sandwich, the automorphism symmetry is non-invertible.}
    \label{fig:sandwichautomorphism}
\end{figure}

We can express the automorphism symmetry on the lattice as
\begin{equation}
    \tilde U_{\rho;K}=\mathcal{D}^\dagger \cdot U_{\rho;K}\cdot \mathcal{D}~,
\end{equation}
where $\mathcal{D}$ is a gapped domain wall acting on the $G$ gauge theory by the projection to $K$ gauge theory subspace, which can be  
obtained by condensing the electric charges that break the gauge group $G$ to $K$; ${\cal D}^\dag$ is gauging the $G/K$ coset symmetry as discussed in \cite{Hsin:2024aqb,Hsin:2025ria}.
On the lattice, the operator $\mathcal{D}$ is obtained by the projection of $G$ gauge fields onto the $K$ gauge fields followed by the projection onto the $K$ Gauss law: the projection is regarded as a sequence of the measurements followed by post selection. 
$U_{\rho;K}$ is the invertible symmetry in $K$ gauge theory. 
In general, $U_{\rho;K}'$ is a defect $V_{\rho;K}$ inducing the automorphism $\rho$ decorated with the gauged SPT symmetry $e^{i\int \alpha}$ due to the twist:
\begin{equation}
    \tilde U_{\rho;K}=\mathcal{D}^\dagger\cdot\left(e^{i\int\alpha}V_{\rho;K}  \right)\cdot \mathcal{D}~,
     \label{eq:sandwich}
\end{equation}
We will give explicit examples in Sec.~\ref{subsec:lattice_sandwich}.

To illustrate the construction, let us consider the special case
$G=H\ltimes K$, and the automorphism is the inner automorphism given by the action of $H$.
The domain wall from $K$ to $G$ gauge theory is obtained by gauging the $H$ symmetry. As a dual description, $K$ gauge theory is obtained by condensing $\text{Rep}(H)$ Wilson lines. In that case, the fusion rule of $\mathcal{D}$ is given by
\begin{align}
    \mathcal{D}\times\mathcal{D}^\dagger = \sum_{h\in H} U_{\rho(h);K}~, \quad \mathcal{D}^\dagger\times \mathcal{D} = \mathcal{C}_{\text{Rep}(H)}~,
\end{align}
where $U_{\rho(h);K}$ is an automorphism symmetry of $K$ gauge theory induced by $h\in H$. $\mathcal{C}_{\text{Rep}(H)}$ is a condensation defect formed by the $\text{Rep}(H)$ Wilson lines. Using the above fusion of $\mathcal{D}$, the fusion rule of the non-invertible automorphism symmetries is obtained by
\begin{align}
    \tilde{U}_{\rho_1;K}\times \tilde{U}_{\rho_2;K} = \sum_{h\in H} \mathcal{D}^\dagger(U^{-1}_{\rho(h);K}U_{\rho_1;K} U_{\rho(h);K} U_{\rho_2;K})\mathcal{D}~,
\end{align}
where we also used $\mathcal{D}^\dagger$ absorbs $U_{\rho(h);K}$ defects. Note that the above fusion of automorphism defects $UUUU$ in $K$ gauge theory generally produces a $K$ gauged SPT defect combined with $U_{\rho(h)\rho_1\rho(h)\rho_2}$, as discussed in Sec.~\ref{subsec:automorphism extended by SPT}.

We note that in the above sandwich construction for the automorphism symmetry, since the $G$ gauge theory is twisted with topological action given by the group cocycle $\omega$, the $K$ gauge theory in the middle is also in general twisted with topological action $\iota^*\omega$ where $\iota:K\rightarrow G$ is the inclusion map. In addition, on the condensation domain wall that breaks $G$ to $K$ there is also topological action given by anomaly inflow: the topological action for $\omega$ restricted to the $G$ gauge field $a$ that takes the form of $K$ gauge field transformed by $G$-valued scalar $\phi$, $a= a|_K^{\phi}$, is given by $\iota^*\omega(a|_K)$ together with boundary terms. These boundary terms are topological actions of $\phi,a$ on the domain wall. Mathematically, it is some kind of modified transgression of the group cocycle $\omega$: when $K=1$ this is the transgression of $\omega$ (for a recent discussion of transgression in the study of symmetries, see \cite{Feng:2025yge}).

\subsubsection{Lattice models with non-invertible automorphism symmetries}
\label{subsec:lattice_sandwich}

The lattice model for twisted gauge theories with cocycle $\omega$ can be constructed in two steps: (1) first, use the standard method for constructing fixed-point lattice model for $G$ SPT for group $\omega$ as discussed in \cite{Chen:2011pg}, (2) gauge the $G$ symmetry to obtain twisted $G$ gauge theory with topological action $\omega$.

As an example, we consider a 2+1d model with $G=(\Z_2)^3,$ and
\begin{align}
    \omega = \pi(a_1\cup a_1\cup a_2+a_1\cup a_2\cup a_3 + a_1\cup a_1\cup a_3)~,
\end{align}
with the $\Z_2$ gauge fields $a_1,a_2,a_3$. The twisted $\Z_2^3$ gauge theory is described by a Clifford stabilizer model on a square lattice (or more generally, a triangulation) on a 2d space.

Each edge of a square lattice has three qubits $\{Z_i,X_i\}$ with $i=1,2,3$.
The stabilizer Hamiltonian is given by
\begin{align}
    H= H_{\text{Gauss}} + H_{\text{flux}}~,
\end{align}
with
\begin{align}
\begin{split}
    H_{\text{Gauss}} &= -\sum_v \left( (-1)^{\int \hat v a_1a_2 + a_1\hat{v}a_2+\hat v d\hat v a_2+ \hat v a_2a_3 +\hat v a_1a_3 + a_1\hat{v}a_3+\hat v d\hat v a_3}\prod_{\partial e\supset v} X_{1,e}\right) \\
    &= -\sum_v \left( (-1)^{\int a_1a_1\hat v  + a_1\hat v a_3}\prod_{\partial e\supset v} X_{2,e}\right) \\
    &= -\sum_v \left( (-1)^{\int a_1a_2\hat v + a_1a_1\hat v}\prod_{\partial e\supset v} X_{3,e}\right)~, \\
\end{split}
\end{align}
\begin{align}
    H_{\text{flux}} = -\sum_{f}\left(\prod_{e\in \partial f} Z_{1,e}\right)-\sum_{f}\left(\prod_{e\in \partial f} Z_{2,e}\right)-\sum_{f}\left(\prod_{e\in \partial f} Z_{3,e}\right)~,
\end{align}
where $\hat{v}$ is a 0-cochain which becomes 1 on a single vertex $v$, otherwise 0.
The integral appearing in the Hamiltonian $H_{\text{Gauss}}$ is over the whole space, but due to locality of $\hat{v}$ this is a local Hamiltonian consisting of $CZ, Z$ operators.

Let us describe the non-invertible automorphism symmetry of this model.
We consider the automorphism swapping the $\Z_2$ gauge fields as
\begin{align}
    a_1\leftrightarrow a_2~,
\end{align}
which corresponds to the automorphism acting on group generators by $(g_1,g_2,g_3)\leftrightarrow (g_2,g_1,g_3)$. This automorphism does not leave $\omega$ invariant, therefore this does not generate the symmetry of the twisted $\Z_2^3$ gauge theory. Meanwhile, this automorphism is realized by a non-invertible symmetry by setting $K=\Z_2^2$, where the condensation $G\to K$ is performed by condensing the diagonal electric particles $e_3$. 

This condensation enforces $a_3=0$ to $\Z_2$ gauge fields, therefore the resulting theory after the condensation becomes $K=\Z_2^2$ gauge theory with the twist
\begin{align}
    \omega|_K = \pi a_1\cup a_1\cup a_2~,
\end{align}
which has an automorphism symmetry $a_1\leftrightarrow a_2$. This automorphism is generated by
\begin{align}
    U = (-1)^{\int \frac{a_1\cup a_2}{2} + a_1\cup_1 \frac{d\tilde a_2}{2}}\prod_{e}(\text{SWAP}_{1,2})_e~,
\end{align}
where $\tilde a$ is $\Z_4$ lift of the $\Z_2$ gauge field $a$. The operator $U$ becomes $\Z_2$ symmetry of $\Z_2^2$ gauge theory, $U^2=1$. 

The condensation of $e_3$ is performed by the following projection into the states with $a_3=0$,
\begin{align}
    \mathcal{D}_{3}= \prod_e \frac{1+Z_{3,e}}{2}~.
\end{align}
Then, the non-invertible automorphism symmetry is generated by
\begin{align}
    \tilde U = \Pi \mathcal{D}_3  U \mathcal{D}_{3}\Pi~,
\end{align}
where $\Pi$ is a projection onto the ground states.

Then, the fusion rule of these symmetry operators is given by
\begin{align}
\begin{split}
    \tilde U\times \tilde U &= \Pi \mathcal{D}_3  U (\mathcal{D}_{3}\Pi\mathcal{D}_3) U\mathcal{D}_3\Pi \\
    &= \Pi \mathcal{D}_3  U \left({1 +(-1)^{\int a_1a_2 + a_1a_1}} \right) U\mathcal{D}_3\Pi \\
    &=  \left({1 +(-1)^{\int a_1a_2 + a_2a_2}} \right)\Pi  \mathcal{D}_3\Pi \\
    &=  \mathcal{C}_3 + (-1)^{\int a_1a_2 + a_2a_2}\times \mathcal{C}_3~,
\end{split}
\end{align}
where $\mathcal{C}_3$ is a condensation operator of the $e_3$ Wilson line,
\begin{align}
    \mathcal{C}_3\propto \sum_{\gamma\in H_1(M,\Z_2)} \left(\prod_{e\subset \gamma} Z_{3;e}\right)\times \Pi~.
\end{align}
Therefore, the fusion of the non-invertible automorphism $\tilde U$ splits into the operators $\mathcal{C}_3$ and its combination with the gauged SPT defect $(-1)^{\int a_1a_2 + a_2a_2}$. This gauged SPT defect is a nontrivial $\Z_2$ symmetry of $\Z_2^3$ gauge theory.

\subsubsection{Example: $U(1)\times U(1)$ Maxwell theory in 3+1d with mixed theta term}
\label{subsub:maxwell}

Consider $U(1)\times U(1)$ Maxwell gauge theory in 3+1d with the mixed theta term
\begin{equation}
    \omega=\frac{\theta}{(2\pi)^2}da_1 da_2~,
\end{equation}
where $a_1,a_2$ are the two $U(1)$ gauge fields, and $\theta\sim \theta+2\pi$.

The gauge group $U(1)\times U(1)$ has the automorphism
\begin{equation}
    \rho:\quad (g_1,g_2)\rightarrow (g_1+g_2,g_2) \ \text{ mod }2\pi~,
\end{equation}
where the $U(1)$ elements are expressed as the phase $e^{ig}$ with $g\sim g+2\pi$. The automorphism transforms the gauge fields by $a_2\to a_2+a_1$, and generates $\mathbb{Z}$ group.
Under the automorphism $\rho$, the topological action changes by a fractional Chern-Simons term:
\begin{equation}
    \omega\rightarrow \omega+d\alpha,\quad \alpha= \frac{\theta}{4\pi^2}a_1da_1~.
\end{equation}
Thus the automorphism symmetry is decorated with fractional quantum Hall (FQH) state:
\begin{equation}
    U_\rho=(\text{Fractional quantum Hall state with }\sigma_{xy}=\theta/\pi)V_\rho~.
\end{equation}
There are three cases:
\begin{itemize}
    \item When $\theta=0$ mod $2\pi$, the FQH state is absent, or can be removed by stacking with bosonic integer quantum Hall (IQH) states.
The automorphism symmetry is a 0-form $\mathbb{Z}$ symmetry that does not participate in higher symmetries.

    \item When $\theta=\pi$ mod $2\pi$, the FQH state is a fermionic integer quantum Hall state. Since the theory is bosonic, the fermionic integer quantum Hall state cannot be removed; in such case, the automorphism symmetry is mixed with the magnetic one-form symmetry to become a 2-group symmetry. To see this, we note that the automorphism $\rho$ changes the action by
    \begin{equation}
        \frac{\pi}{4\pi^2}\int da_1 da_1=\pi\int \frac{da_1}{2\pi}\frac{da_1}{2\pi}=\pi\int \frac{da_1}{2\pi}\cup (w_2+w_1^2)=\int \frac{da_1}{2\pi}\Delta B^M~,
    \end{equation}
    where $\Delta B^M=\pi (w_2+w_1^2)$ is the change of the background 2-form for the $U(1)$ magnetic 1-form symmetry with conserved charge $\oint \frac{da_1}{2\pi}$. Since the automorphism symmetry changes the background of a 1-form symmetry, they become a 2-group symmetry \cite{Benini:2018reh,Cordova:2018cvg}.

\item When $\theta=2\pi \frac{p}{q}$ for co-prime integers $p,q$ and $\theta\neq 0,\pi$ mod $2\pi$, the fractional quantum Hall state on the automorphism symmetry is non-invertible, and the automorphism symmetry becomes a non-invertible 0-form symmetry. This is similar to the non-invertible 0-form symmetries considered in \cite{Choi:2022jqy,Cordova:2022ieu}.

\end{itemize}

\subsubsection{Example: $\mathbb{Z}_N$ Walker-Wang model in 3+1d}
\label{sec:WWhighergroup}

Consider two-form gauge theory in 3+1d with gauge group $G=\mathbb{Z}_{N}$ and action \cite{Kapustin:2014gua,Gaiotto:2014kfa,Hsin:2018vcg}
\begin{equation}
    2\pi\int \frac{p}{2N}{\cal P}(b)~,
\end{equation}
where $b$ is the 2-form gauge field with holonomy in $\mathbb{Z}_N$. 
The theory can also be described on the lattice via Walker-Wang model \cite{Walker:2011mda,von_Keyserlingk_2013}. 
For any integer $r$ such that $\gcd(r,N)=1$, consider the automorphism $b\rightarrow rb$ mod $N$. The action changes by
\begin{equation}
    2\pi\frac{p(r^2-1)}{2N}\int {\cal P}(b)~.
\end{equation}
\begin{itemize}
    \item When $p(r^2-1)=0$ mod $2N$, the action is invariant and this is a 0-form symmetry that does not participate in higher group, i.e. one can couple the symmetry with background gauge field without activating other backgrounds.

    \item When $p(r^2-1)=N$ mod $2N$, the action can be made invariant by turning the symmetry into a 2-group. 
The change of the action is compensated by shifting the background of one-form symmetry generated by $\int b$:
\begin{equation}
    \frac{2\pi}{N}\int b\cup B~.
\end{equation}
Under the transformation $b\rightarrow b\rightarrow rb$, $B$ transforms as
\begin{equation}
    B\rightharpoonup r^{-1}B+\frac{N}{2} w_2~.
\end{equation}
Note that it is not consistent to turn off $B_2=0$ due to the shift $w_2$, and thus the symmetry is a nontrivial 2-group \cite{Benini:2018reh,Cordova:2018cvg}.

\item When $p(r^2-1)\neq 0,N$ mod $2N$, the action can be made invariant by turning the symmetry into a non-invertible symmetry.
We decorate the domain wall that generates the automorphism symmetry with a TQFT that has $\mathbb{Z}_N$ anyon with topological spin $\frac{p(r^2-1)}{2N}$ mod 1. Since $p(r^2-1)\neq 0,N$ mod $2N$, this is a nontrivial TQFT, and the symmetry becomes non-invertible. 
This is similar to the FQH state used to decorating the chiral symmetry in QED$_4$ discussed in \cite{Choi:2022jqy,Cordova:2022ieu}.

\end{itemize}

\subsection{Automorphism symmetry and 't Hooft anomalies}

The gauged SPT operator that decorate the automorphism symmetry can also result in 't Hooft anomalies mixed with center 1-form symmetry. This is similar to the anomalies in twisted gauge theories discussed in \cite{Barkeshli:2022edm}.

\paragraph{$U(1)$ gauge theory in 1+1d}

As an example, consider $U(1)$ gauge theory in 1+1d, where the $U(1)$ gauge field has theta term $\theta_{2d}=\pi$ for $U(1)$ gauge field $a$:
\begin{equation}
\omega=    \pi\frac{da}{2\pi}~.
\end{equation}
The gauge group has charge conjugation automorphism
\begin{equation}
    \rho: \quad g\rightarrow -g~,
\end{equation}
where the $U(1)$ elements are the phases $e^{ig}$.
Under the automorphism $\rho$, the topological action changes by
\begin{equation}
    \omega\rightarrow \omega+d\alpha,\quad \alpha=a~.
\end{equation}
Thus the automorphism symmetry is
\begin{equation}
    U_\rho=e^{i\int a}V_\rho~.
\end{equation}
Since the automorphism symmetry carries charge 1 under the center 1-form symmetry, we conclude that there is a nontrivial mixed anomaly between the automorphism symmetry and the 1-form symmetry.

Similarly, for $U(1)$ gauge theory coupled to two charge-1 scalars with $SO(3)$ global symmetry (where the scalars transform as the doublet with fractionalized $SO(3)$ symmetry), the charge Wilson line $e^{i\int a}$ carries fractional charge under the $SO(3)$ global symmetry, and thus there is a mixed anomaly between the charge conjugation symmetry and the $SO(3)$ symmetry (for related discussion, see e.g. \cite{Gaiotto:2017yup,Komargodski:2017dmc,Sulejmanpasic:2018upi}).

\subsection{Phase transitions with automorphism symmetry}

While the examples we have discussed so far are pure gauge theories, we can couple the gauge fields to matter in representation $R$ that is invariant under the automorphism of the gauge group. The discussions still apply to these gauge theories with matters. The topological term of the gauge fields can also result from fermion matter fields (see e.g. \cite{Alvarez-Gaume:1984zst}), where weakly coupled fermions can give rise to nontrivial topological responses as described by free-fermion SPT phases (e.g. \cite{Wen_2012,Witten:2015aba}).

Consider $U(1)\times U(1)$ gauge theory in 2+1d coupled to two scalars $\Phi_1,\Phi_2$ with charges $(q_1,q_2)$ and $(q_2,q_1)$ with $q_1\neq q_2$ under the two $U(1)$ gauge groups. The gauge theory has a mixed Chern-Simons term
\begin{equation}
    \omega=\frac{k}{2\pi}a_1 da_2~,
\end{equation}
where $a_1,a_2$ are the two $U(1)$ gauge fields. The theory has SWAP symmetry for the swap automorphism $\rho$ that exchanges the two $U(1)$s. The swap automorphism exchanges $\Phi_1\leftrightarrow \Phi_2$ and
 $a_1\leftrightarrow a_2$. Under the transformation, the mixed Chern-Simons term changes as
\begin{equation}
    \omega\rightarrow \frac{k}{2\pi}a_2da_1=\omega +d\alpha,\quad \alpha=\frac{k}{2\pi}a_2a_1~.
\end{equation}
 Thus the automorphism symmetry is
 \begin{equation}
     U_\rho=e^{\frac{ik}{2\pi}\int a_2a_1}V_\rho~.
 \end{equation}
We note that the automorphism symmetry still has order 2:
\begin{equation}
    U_\rho(M)\cdot U_\rho(M)=e^{\frac{ik}{2\pi}\int_M a_2a_1}V_\rho\cdot e^{\frac{ik}{2\pi}\int_M a_2a_1}V_\rho=e^{\frac{ik}{2\pi}\int_M (a_2a_1+a_1a_2)}=1~.
\end{equation}

In addition to the automorphism 0-form symmetry, the theory can also have 1-form symmetry, which shifts the gauge fields $a_1,a_2$ by flat connections $a_1\rightarrow a_1+\lambda_1$, $a_2\rightarrow a_2+\lambda_2$.
The mixed Chern-Simons term changes by
\begin{equation}
    \frac{k}{2\pi}\int a_1 da_2\rightarrow
    \frac{k}{2\pi}\int a_1 da_2+\frac{k}{2\pi}\int (\lambda_1 da_2+da_1\lambda_2+\lambda_1 d\lambda_2)~.
\end{equation}
On the right hand side, the last term is an 't Hooft anomaly. The other terms $\lambda_1 da_2+da_1 \lambda_2$ breaks the 1-form symmetry explicitly to $\mathbb{Z}_k\times \mathbb{Z}_k$, where $\lambda_1,\lambda_2$ have holonomies in $\frac{2\pi}{k}\mathbb{Z}$:
\begin{equation}
    \lambda_1=\frac{2\pi m_1}{k},\quad \lambda_2=\frac{2\pi m_2}{k},\quad m_1,m_2\in\mathbb{Z}_k~.
\end{equation}
Furthermore, the electric charges of the matter fields break the 1-form symmetry to
\begin{equation}
    q_1 \lambda_1+q_2\lambda_2=0,\quad q_2\lambda_1+q_1\lambda_2=0\quad (\text{mod }2\pi)~.
\end{equation}
In other words, the 1-form symmetry is 
\begin{equation}
{\cal G}^{(1)}:=\left\{(m_1,m_2)\in\mathbb{Z}_k\times\mathbb{Z}_k:q_1m_1+q_2m_2=q_2m_1+q_1m_2=0\text{ mod }k\right\}\subset \mathbb{Z}_k\times\mathbb{Z}_k~.
\end{equation}
For example, if $q_1,q_2$ are both multiple of $k$, then the 1-form symmetry is the same as $\mathbb{Z}_k\times\mathbb{Z}_k$. 
The 1-form symmetry has a diagonal subgroup $m_1=m_2\in \frac{k}{\gcd(q_1+q_2,k)}$ invariant under the swap automorphism.

We note that due to the gauged SPT defect decorating the automorphism symmetry, under 1-form transformation that shifts the gauge fields $a_1,a_2$ the automorphism symmetry acquires a phase, which indicates a mixed 't Hooft anomalies of the symmetries.

The scalars can have a quartic potential $V(\Phi_1,\Phi_2)=V(\Phi_2,\Phi_1)$. When the parameters in the potential takes values such that scalars condense with $\langle \Phi_1\rangle\neq \langle \Phi_2\rangle$, the swap automorphism symmetry is spontaneously broken. On the other hand, when the parameters are such that the scalars do not condense\footnote{
The potential can be constructed using the method in e.g. \cite{Cordova:2018qvg}.
}, the theory flows to a gapped phase with the mixed Chern-Simons term that describes untwisted $\mathbb{Z}_k$ gauge theory enriched with unbroken  automorphism symmetry that becomes the electro-magnetic duality symmetry. In this phase, the 1-form symmetry is enlarged from ${\cal G}^{(1)}$ to $\mathbb{Z}_k\times\mathbb{Z}_k$.
The theory then realizes a symmetry-breaking phase transition for the automorphism symmetry. In particular, the automorphism symmetry can have mixed anomalies with the center 1-form symmetry. For the special case $k=2$, the anomaly is discussed in \cite{Barkeshli:2021ypb}.

\section{Higher-group symmetry from interactions with gauged SPT symmetries}
\label{sec:higher group gauged SPT}

\begin{figure}[t]
    \centering
    \includegraphics[width=0.25\linewidth]{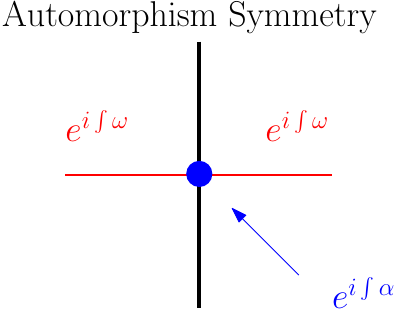}
    \caption{Junction of automorphism symmetry (black) in untwisted gauge theory with gauged SPT symmetry (red) for cocycle $\omega$ of the gauge group. For automorphism that changes $\omega$ by $d\alpha$, the junction (blue) is decorated with $e^{i\int\alpha}$.}
    \label{fig:junction}
\end{figure}

While the discussion so far has focused on automorphism symmetries in twisted gauge theories, we can extend the discussion to the interaction between automorphism symmetry and gauged SPT symmetries \cite{Hsin:2019fhf,barkeshli2023codimension,Barkeshli:2022edm,Barkeshli:2023bta,Hsin:2024nwc}, i.e. twist on submanifolds instead of the entire spacetime. This leads to higher-group symmetry involving the automorphism and gauged SPT symmetries in generic gauge theory, including untwisted ones.

For simplicity, we will focus on untwisted gauge theories in the following; the discussion can generalize to twisted gauge theories as well, where some of the gauged SPT symmetries become trivial \cite{Barkeshli:2022edm}.

Consider gauged SPT symmetry with group cocycle $\omega$ intersecting the automorphism symmetry as shown in Fig.~\ref{fig:junction}. 
Locally at the intersection, it is automorphism symmetry applied on the submanifold that supports the gauged SPT defects.
There are several cases where automorphism symmetries form higher-group symmetry:
\begin{enumerate}
    \item 
    Let us consider a $k$-dimensional gauged SPT symmetry of untwisted $G$ gauge theory, described by $[\omega]\in H^k(BG,U(1))$.
    If the automorphism preserves the gauged SPT symmetry $[\rho^*\omega]=[\omega]$, the automorphism at the intersection can be decorated with additional $(k-1)$-dimensional gauged SPT defect $e^{i\int \alpha}$, with $\alpha$ satisfying $d\alpha=\rho^*\omega-\omega$; at the intersection with the SPT symmetry defect $\omega$, the automorphism defect has the form $e^{i\int \alpha}V_\rho$. 
    As seen in Sec.~\ref{subsec:automorphism extended by SPT}, such a decorated automorphism symmetry is in general extended by the gauged SPT symmetry; the fusion of the operators $e^{i\int \alpha}V_\rho$ at the defect $\omega$ produces the $(k-1)$-dimensional gauged SPT symmetry $e^{i\int \eta}$.
    This implies a higher-group symmetry \cite{Benini:2018reh}: junction of automorphism symmetries intersecting the gauged SPT defect with cocycle $\omega$ in $k$ dimensions produces a gauged SPT defect $e^{i\int \eta}$ supported on a lower $(k-1)$-dimensional submanifold. See Fig.~\ref{fig:fusing}.
    We note that while $\alpha$ is not a cocycle, $\eta$ is a cocycle and thus the intersects emits a well-defined gauged SPT symmetry defect that generates a higher-form symmetry.

\item Similarly, if we fuse the $k$-dimensional gauged SPT defects $e^{i\int\omega}$ in the presence of the automorphism defects, this emits a $(k-1)$-dimensional gauged SPT defect $e^{i\int \beta}$. 

Consider an automorphism $\rho:G\to G$ and gauged SPT symmetry with cocycle $[\omega]\in H^k(BG,U(1))$ on $k$-dimensional submanifold $M$. 
Suppose that the cohomology class $[\omega]$ is invariant under the automorphism $\rho$, then one introduces a $(k-1)$-cochain induced by the automorphism $\alpha\in C^{k-1}(BG,U(1))$, satisfying
\begin{align}
    d\alpha = \rho^*\omega-\omega~.
\end{align}

The gauged SPT symmetry is a cyclic Abelian group $\mathbb{Z}_n$, where $e^{in\int_M\omega}=1$. 
We can take a cocycle representative $\omega$ such that $n\omega=0$ mod $2\pi$, then $n\omega$ is invariant under automorphism. This implies that $nd\alpha=0$ mod $2\pi$, and thus the cochain $\alpha$ is a $1/n$ fraction of a discrete cocycle,
\begin{equation}
    \alpha=\frac{1}{n}\beta~,
    \label{eq:betaformula}
\end{equation}
with $[\eta]\in H^{r-2}(BG,U(1))$.
Consider the junction of $e^{i[n_1]\int\omega},e^{i[n_2]\int \omega}$ and $e^{i[n_1+n_2]\int\omega}$, where $[n_1]$ denotes the restriction to $0,1,\cdots,n-1$.
When the junction intersects with the automorphism symmetry, this also induces a junction of $e^{i[n_1]\int \alpha}$, $e^{i[n_2]\int \alpha}$, $e^{i[n_1+n_2]\int \alpha}$, which differs by a gauged SPT symmetry that emits from the junction
\begin{equation}
    e^{i\frac{[n_1]+[n_2]-[n_1+n_2]}{n}\int \beta}~.
\end{equation}
Thus the $(D-k)$-form symmetry generated by $e^{i\int \beta}$ together with automorphism 0-form symmetry and $(D-k-1)$-form symmetry forms a higher group.

\item We have seen that the interaction between an automorphism and gauged SPT symmetries produces a gauged SPT defect $e^{i\int\eta}$ or $e^{i\int\beta}$.
This gauged SPT operator further forms a higher-group symmetry together with magnetic defects of gauge theory~\cite{Barkeshli:2022edm}. Namely, the commutation relation between $e^{i\int\eta}$ (or $e^{i\int\beta}$) and magnetic defect operators produces a lower-dimensional gauged SPT symmetry, which is again regarded as a higher-group symmetry. Such an example is found in e.g., $\Z_N\times\Z_M$ gauge theory in generic dimensions, as described in Sec.~\ref{subsubsec:ZNZM in generic dim}.

    \item If the automorphism does not preserve the gauged SPT symmetry $[\rho^*\omega]\neq [\omega]$, then the automorphism permutes the gauged SPT symmetry $e^{i\int \omega}$ into $e^{i\int \rho^*\omega}$. As discussed previously, the intersection is a non-invertible symmetry on the submanifold.
    
    On the other hand, if the junction further intersects with the condensation defects that break $G$ to subgroup $K$ such that $[\omega]$ restricted to $K$ is preserved by the automorphism, then the automorphism can again act projectively at the intersection as above to give a higher-group symmetry.
    
\end{enumerate}

\begin{figure}[t]
    \centering
    \includegraphics[width=0.85\linewidth]{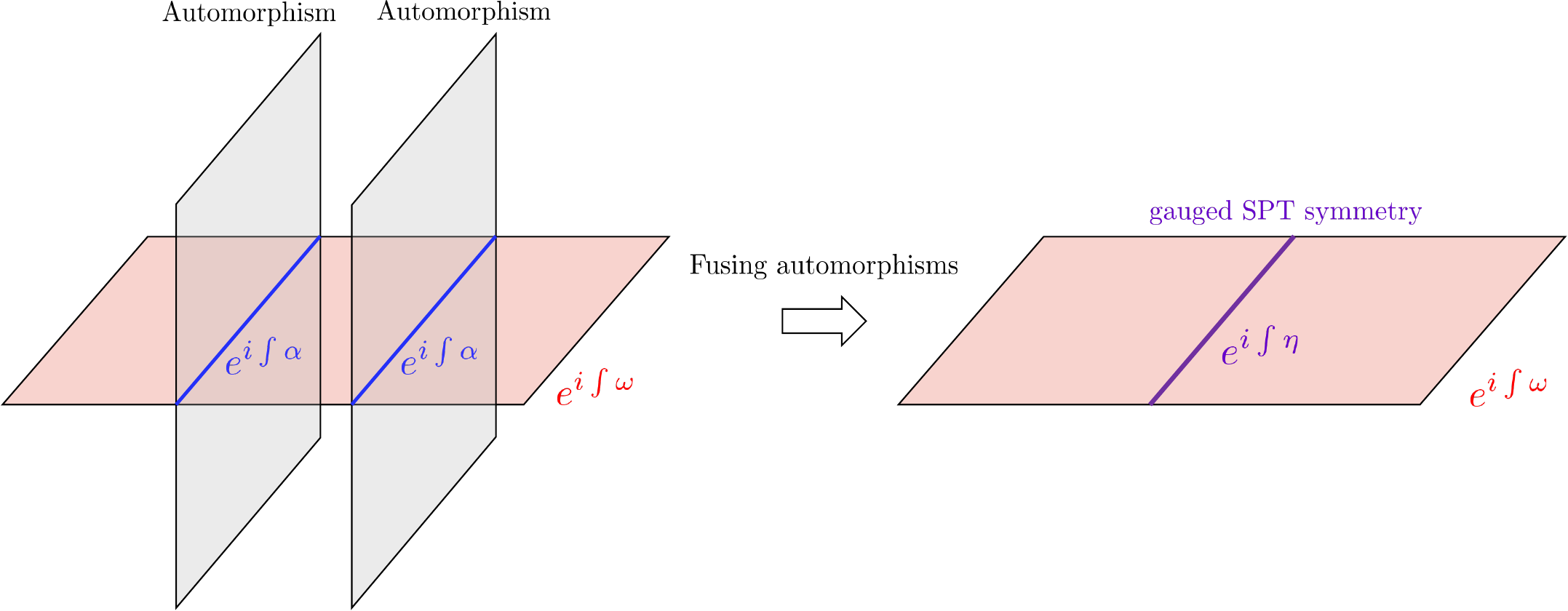}
    \caption{Fusing automorphism symmetry defects in the presence of a gauged SPT symmetry defect $\omega$ in $k$ dimensions. This leaves a $(k-1)$-dimensional gauged SPT symmetry $\eta$, regarded as a higher-group symmetry.}
    \label{fig:fusing}
\end{figure}

\subsection{Example: $\mathbb{Z}_2\times\mathbb{Z}_2$ gauge theory in 3+1d}
\label{subsubsec:Z2Z2 in 3+1d}

Let us focus on the second scenario for getting the higher-group symmetry described above.
We can characterize this higher-group structure concretely using operators on a lattice model as follows. Let us consider a twisted Hilbert space on a space in the form of $M^{d-1}\times S^1_\rho$, where $S^1_\rho$ direction is twisted by the automorphism symmetry $\rho$. Then, the gauged SPT symmetry $U_\omega$ along a $k$-submanifold $M_{k-1}\times S^1_\rho$ has the form of
\begin{align}
    U_\omega(M_{k-1}\times S^1_\rho) = e^{i\int_{M_{k-1}\times S^1_\rho}\omega}\times e^{i\int_{M_{k-1}}\alpha}~.
\end{align}
Suppose that $U_\omega$ has order $n$, then the $n$-th power of $U_\omega$ at the twisted Hilbert space is given by
\begin{align}
    (U_\omega(M_{k-1}\times S^1_\rho))^n = e^{i\int_{M_{k-1}}\beta}~,
\end{align}
with $\beta$ is given by Eq.~\eqref{eq:betaformula}. 

An example of such a higher-group symmetry is found in untwisted $\Z_2\times\Z_2$ gauge theory in 3+1d.
Let us consider the theory on a 3d torus $T^3$ space, and insert the automorphism defect along $T^2_{xy}$ that exchanges two $\Z_2$s, so that we get a twisted boundary condition along $S^1_z$. This defect Hamiltonian has a non-Clifford logical gate generated by the 3d gauged SPT operator with $\omega = \frac{\pi}{2}a_1 da_2 $,
\begin{align}
    U_\omega(T^3) = e^{i\int_{T^3} \omega} e^{i\int_{T^2_{xy}}\alpha}~,
\end{align}
where $\alpha = \frac{\pi}{2}(a_1a_2 + da_1\cup_1 a_2)$. The integral of $\alpha$ on $T_{xy}^2$ generates the Controlled-$S$ gate. 

In particular, the fusion of two gauged SPT defects $U_\omega$ at the automorphism defect generates the 1+1d gauged SPT defect
\begin{align}\label{eqn:highergroupexample}
    e^{i\pi \int_{T^2_{xy}} a_1a_2}
\end{align}

The equation (\ref{eqn:highergroupexample}) implies the non-trivial higher group structure between 0-form symmetry and 1-form symmetry. This effect is given by the equations among background gauge fields
\begin{align}
    dC_2 = A_1 \frac{d\hat C_1}{2}~,
\end{align}
where $C_2$, $C_1$ are the backgrounds of 1+1d, 2+1d gauged SPT defects, and $A_1$ is the background of the automorphism defect.

\subsection{Example: $\mathbb{Z}_N\times\mathbb{Z}_M$ gauge theory in general spacetime dimension}
\label{subsubsec:ZNZM in generic dim}

Let us focus on the third scenario for getting the higher-group symmetry.
Consider untwisted $\mathbb{Z}_N\times\mathbb{Z}_M$ gauge theory in $D$ spacetime dimension. The theory has gauged SPT $(D-4)$-form symmetry supported on 3-manifold $M_3$:
\begin{equation}
    e^{\frac{2\pi i k}{N}\int_{M_3}  a_1\cup \frac{da_2}{M}}~,
    \label{eq: SPT in ZN x ZM}
\end{equation}
where $a_1,a_2$ are the $\mathbb{Z}_N,\mathbb{Z}_M$ gauge fields, respectively. We note that the operator has order $\ell:=\gcd(M,N)$: multiplying the topological term by $\ell$ gives: (we can express $\gcd(N,M)=q N+r M  $ for integers $q,r$) 
\begin{align}
    &\frac{2\pi k(q N+r M)}{N} a_1\cup \frac{da_2}{M}
    =2\pi kq a_1\cup \frac{da_2}{M}+\frac{2\pi kr }{N} a_1\cup da_2\cr
    &\quad =2\pi kq a_1\cup \frac{da_2}{M}
    +\frac{2\pi kr }{N} da_1\cup a_2
    -\frac{2\pi kr }{N} d\left(a_1\cup a_2\right)\cr 
    &\quad =-\frac{2\pi kr }{N} d\left(a_1\cup a_2\right)\text{ mod }2\pi=d\alpha \text{ mod }2\pi~.
\end{align}

The automorphism we will consider can be defined for the cases where $\gcd(1+\ell,N)=1$. This is equivalent to $\gcd(1+\ell,N)=\gcd(1+\ell,N-\frac{N}{\ell}(1+\ell))=\gcd(1+\ell,\frac{N}{\ell})$. For such cases, there is an automorphism
\begin{equation}
    \rho:(n,m)\in\mathbb{Z}_N\times\mathbb{Z}_M \; \rightarrow \; \left((1+\ell)n,m\right)~.
\end{equation}
At the intersection with the operator \eqref{eq: SPT in ZN x ZM}, the automorphism is decorated with the 2d gauged SPT operator
\begin{equation}
    e^{i\int \alpha}=e^{-\frac{2\pi i kr}{N}\int a_1\cup a_2}~.
\end{equation}

When the magnetic operator for the $a_2$ gauge field with flux $m$, which has dimension $(D-2)$ and generates a $\mathbb{Z}_M$ 1-form symmetry, intersects the domain wall that generates the automorphism 0-form symmetry, the intersection has
\begin{equation}
    e^{-\frac{2\pi i km r}{N}\int a_1}~.
\end{equation}
The above Wilson line operator generates a $(D-2)$-form symmetry. Thus the automorphism 0-form symmetry, the gauged SPT $(D-4)$-form symmetry, the 1-form symmetry, and the $(D-2)$-form symmetry mix to become a $(D-1)$-group symmetry for $D\geq 4$. The higher group symmetry can be described in terms of their background fields $C_1,B_{D-3},A_2,B_{D-1}$ respectively: from the junction of defects, we find the relation \cite{Benini:2018reh}
\begin{equation}
    dB_{D-1}=-kr A_2\cup B_{D-3}\cup C_1~.
\end{equation}

\section{Automorphism Symmetries in $\mathbb{Z}_N^m$ Gauge Theories}
\label{sec:twistedZNm}

In this section we study the automorphism symmetries in twisted $\mathbb{Z}_N^m$ gauge theories. We will discuss various automorphisms and study the properties of the corresponding automorphism symmetries.
We note that examples of such automorphism symmetries have been discussed in \cite{Hsin2024_non-Abelian,Kobayashi:2025cfh}.

\subsection{Automorphism Symmetry in $\mathbb{Z}_2^N$ Gauge Theory}
\label{sec:Z2general}

Consider $\mathbb{Z}_2^m$ 1-form gauge Theory in $D$ spacetime dimension. The topological action for the gauge fields $\{a_i\}_{i=1}^m$ is
\begin{equation}
    \pi\int a_{i_1} \cup a_{i_2}\cup\cdots a_{i_D}~, 
\end{equation}
where $a_i,a_j$ can be the same gauge fields. For example, the Levin-Gu topological response is $\pi \int a^3$ for a single $\mathbb{Z}_2$ gauge field $a$.
We will study what becomes of the automorphism of $\mathbb{Z}_2^m$ in the presence of the topological action.

The automorphism of the group $\mathbb{Z}_2^m$ is given by linear transformations of the $N$ generators:
\begin{equation}
    \text{Aut}(\mathbb{Z}_2^m)=SL(m,\mathbb{Z}_2)~.
\end{equation}
For example, the group $\mathbb{Z}_2\times\mathbb{Z}_2$ has 3 nontrivial element, and the automorphism of the group is isomorphic to the $S_3$ permutation.

\subsubsection{Example: $\mathbb{Z}_2^3$ gauge theory with $\mathbb{Z}_7$ symmetry}

Consider twisted $\mathbb{Z}_2^3$ gauge theory in 2+1d with the following topological action: (here, $a_i$ for $i=1,2,3$ are the gauge fields for the three $\mathbb{Z}_2$s in the gauge group)
\begin{equation}
\omega=    \pi\left(a_1\cup a_2\cup a_3+ a_1^2\cup a_2 + a_3^2\cup a_2 + a_1^2\cup a_3 +\sum_{i} a_i^3\right)~.
\end{equation}
Namely, the topological action is obtained by summing over all type I, II, III cocycles of $\Z_2^3$.
 We will focus on the following automorphism of the $\mathbb{Z}_2^3$ gauge group: in terms of the gauge field, it acts as
\begin{equation}
    \rho:\quad a_1\rightarrow a_3,\quad a_2\rightarrow a_1+a_3,\quad a_3\rightarrow a_2~.
\end{equation}
One can verify that the automorphism has order 7:
\begin{align}
\rho^7:\quad     (a_1,a_2,a_3)&\rightarrow (a_3,a_1+a_3,a_2)\rightarrow (a_2,a_2+a_3,a_1+a_3)\rightarrow (a_1+a_3,a_1+a_2+a_3,a_2+a_3)\cr
    &\rightarrow (a_2+a_3,a_1+a_2,a_1+a_2+a_3)\rightarrow (a_1+a_2+a_3,a_1,a_1+a_2)\rightarrow(a_1+a_2,a_3,a_1)\cr 
    &\rightarrow (a_1,a_2,a_3)~.
\end{align}

The action transforms under the automorphism $\rho$ as
\begin{align}
    \omega\to\omega + d\alpha~,
\end{align}
where
\begin{equation}
    \alpha= \pi\left( \frac{a_1\cup a_3}{2}+\frac{a_2\cup a_1}{2} + a_2\cup_1(a_3^2+a_1^2)+a_3\cup_1 (a_1^2+a_1a_2) + a_1\cup (a_1\cup_1a_3) + (a_1\cup_1a_3)\cup a_3 \right)~.
\end{equation}
The automorphism symmetry is $U_\rho=e^{i\int \alpha}V_\rho$.

We note that since the automorphism symmetry is invertible, we can ask whether the symmetry is an extension or higher group for the $\mathbb{Z}_7$ automorphism. Since $\gcd(7,2)=1$ and the automorphism does not have fixed points, there is no nontrivial extension structure. Thus the automorphism symmetry is simply a $\mathbb{Z}_7$ 0-form symmetry.

\subsection{Automorphism symmetry in $\mathbb{Z}_N^m$ gauge theory}

\subsubsection{SWAP automorphism}

Consider $\mathbb{Z}_N^m$ $r$-form gauge theory with action
\begin{equation}
\frac{2\pi}{N}\int  a_1\cup a_2\cup\cdots a_m~.
\end{equation}
The spacetime dimension is $D=rm$.

Consider the swap symmetry that exchanges 
\begin{equation}
    \rho:\quad (a_1,a_2)\rightarrow (a_2,(-1)^{r}a_1)~.
\end{equation}
Using the identity
\begin{align}\label{eqn:swappermutationdecorZN-0}
    &a_1\cup a_2=(-1)^{r}a_2\cup  a_1 -d\left(a_1\cup_{1}a_2\right)+da_1\cup_{1}a_2+(-1)^{r}a_1\cup_{1}da_2\cr &=(-1)^{r}a_2\cup  a_1 -d\left(a_1\cup_{1}a_2\right)\text{ mod }N~
\end{align}
we find that after the symmetry, the action maps to
\begin{align}\label{eqn:swappermutationdecorZN-1}
    &a_1\cup a_2\cdots a_n\; \rightarrow\; (-1)^{r}(a_2\cup a_1)\cup\cdots a_n\cr 
    &\qquad\qquad\;\;\;\;\quad\quad\;\; =(a_1\cup a_2)\cup\cdots a_m+d\alpha\text{ mod }N~,\cr
    &\alpha=(-1)^r\left(a_1\cup_1a_2\right)\cup a_3\cup \cdots\cup a_m~.
\end{align}
Thus the automorphism is 
\begin{equation}\label{eqn:swappermutationdecorZN}
    U_\rho(M)=e^{i\int\alpha}V_\rho=e^{\frac{2\pi i (-1)^r}{N}\int_M \left(a_1\cup_1a_2\right)\cup a_3\cup \cdots\cup a_m}V_\rho(M)~,
\end{equation}
where $M$ is the support where we perform the automorphism.

Let us take the square of the automorphism symmetry: on any closed submanifold $M$,
\begin{equation}
    U_\rho(M)^2=
    e^{\frac{2\pi i (-1)^r}{N}\int_M \left(a_1\cup_1a_2\right)\cup a_3\cup \cdots\cup a_m}e^{\frac{2\pi i (-1)^r}{N}\int_M \left((-1)^ra_2\cup_1a_1\right)\cup a_3\cup \cdots\cup a_m}=1~.
\end{equation}
Thus the symmetry is a $\mathbb{Z}_2$ 0-form symmetry.

\subsubsection{Higher group symmetry in untwisted gauge theories in $(mr+1)$ spacetime dimension}

The discussion above for the swap automorphism symmetry can also apply to untwisted gauge theories in $(rm+1)$ spacetime dimension. In this case, we consider the gauged SPT defects $\exp\left(\frac{2\pi i}{N}\int a_1\cup \cdots a_m\right)$ and the automorphism 0-form symmetry. The above discussion implies that this gauged SPT defect is invariant under the automorphism 0-form symmetry. However, the intersection of these defects traps the gauged SPT defect $e^{i\int\alpha}$.

Now, let us further intersect the junction with a magnetic defect.
We will focus on the magnetic defect that induces the same holonomy for $a_1,a_2$: such magnetic defect is invariant under the automorphism.
Around the magnetic defect, we can replace $a_1,a_2$ by $d\vartheta$ where $\vartheta\sim \vartheta+1$ is the (normalized) angular variable around the magnetic defect. In this normalization, $d\vartheta$ has holonomy $0,1$ around a single circle. When such magnetic defect intersects the automorphism symmetry, the intersection induces the gauged SPT defect \cite{Barkeshli:2022edm}
\begin{equation}
    e^{\frac{2\pi (-1)^r}{N}\int d\vartheta \cup a_3\cdots a_m}=
    e^{\frac{2\pi (-1)^r}{N}\int a_3\cup\cdots\cup a_m}~,
    \label{eq:sourced gauged SPT m-auto}
\end{equation}
where we integrated the gauged SPT factor in the automorphism symmetry over a small circle with nontrivial holonomy $\oint d\vartheta=1$.
Thus the junction of (1) 0-form automorphism symmetry, (2) 0-form gauged SPT symmetry $\exp\left(\frac{2\pi i}{N}\int a_1\cup\cdots\cup a_m\right)$, and (3) 1-form symmetry that shifts $a_1,a_2$ at the same time, produces a higher $(2r)$-form symmetry generated by the gauged SPT defect $e^{\frac{2\pi (-1)^r}{N}\int a_3\cup\cdots\cup a_m}$ of codimension $(mr+1)-(m-2)r=2r+1$. This is a $(2r+1)$-group symmetry. See Fig.~\ref{fig:m_auto} for an illustration.

\begin{figure}[t]
    \centering
    \includegraphics[width=0.4\linewidth]{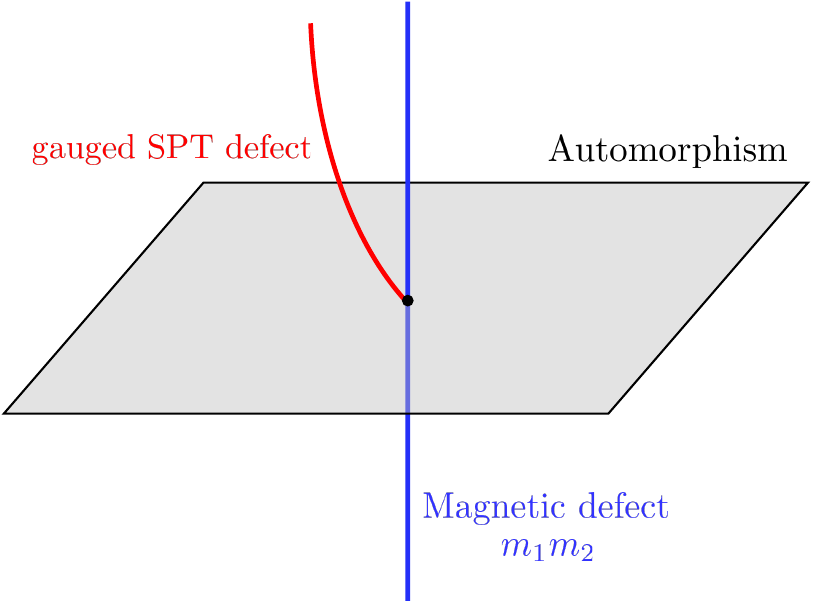}
    \caption{The intersection between the automorphism symmetry defect and the magnetic defect shifting $a_1,a_2$ (which we denote by $m_1m_2$ in the figure) sources a gauged SPT symmetry defect in Eq.~\eqref{eq:sourced gauged SPT m-auto}.}
    \label{fig:m_auto}
\end{figure}

\section{Transversal Non-Clifford Logical Gates from Automorphisms}
\label{sec:noncliffordgate}

In this section, we will construct non-Clifford gates induced by automorphisms of the underlying gauge group of homological codes. 
For the application here, the gauge group $G$ will be copies of $\mathbb{Z}_2$ or $\mathbb{Z}_N$ qudits. The automorphism symmetries $U_\rho$ of the twisted $G$ gauge theory is an emergent symmetry of the lattice model, and they give rise to logical gate in the topological quantum codes.

\subsection{Review of qudit Clifford hierarchy and universal quantum computation}

The Clifford hierarchy \cite{Gottesman1999hierarchy} provides a structured way to classify logical gates of error-correcting codes according
to how they conjugate the Pauli group. It plays a central role in
understanding universality in fault-tolerant quantum computations.

Let $\mathcal{P}_n$ denote the $n$-qubit Pauli group.  Following Ref.~\cite{Gottesman1999hierarchy}, the Clifford hierarchy is defined recursively by
\begin{align}
\mathcal{C}_1 &= \mathcal{P}_n, \\
\mathcal{C}_{k+1}
  &= \left\{\, U \;\middle|\; U \mathcal{P}_n U^\dagger \subset \mathcal{C}_k \right\}.
\end{align}
Thus $\mathcal{C}_2$ is the usual Clifford group, i.e., the normalizer of the
Pauli group. The set $\mathcal{C}_k$ with $k\ge 3$ is not a group.
The hierarchy forms a nested sequence
\[
\mathcal{C}_1 \subseteq \mathcal{C}_2 \subseteq \mathcal{C}_3 \subseteq \cdots.
\]
Level $1$ contains Pauli operators, level $2$ the Clifford gates, and
level $3$ contains gates that map Pauli operators to Clifford operators under conjugation.  A canonical example of a level-$3$ gate on qubits is
the $T$ (or $\pi/8$) gate.

A non-Clifford gate is any unitary $U \notin \mathcal{C}_2$.  Such gates occupy levels $3$ and above in the hierarchy.  They are essential in quantum computation, since Clifford circuits alone can be efficiently classically simulated (the Gottesman-Knill theorem \cite{Gottesman:1998hu}).  To obtain a universal gate set, it is well-known that at least one non-Clifford gate is required:
\begin{theorem}
\cite{bravyi2005}     The full Clifford group together with any single non-Clifford gate
generates a universal (dense) subgroup of $SU(2^n)$.
\end{theorem}

A common example is the Clifford+$T$ gate set for qubits.  Since $T$ is a level-$3$ non-Clifford gate, its addition completes universal gate set.

\paragraph{Clifford Hierarchy for $\mathbb{Z}_N$ qudits}

Consider a $d$-dimensional qudit with computational basis
$\{\lvert 0\rangle, \dots, \lvert d-1\rangle\}$.  The generalized Pauli
operators are
\begin{align}
X \lvert j\rangle &= \lvert j+1 \bmod d \rangle, \\
Z \lvert j\rangle &= \omega^j \lvert j \rangle,
\qquad \omega = e^{2\pi i/d}.
\end{align}
The generalized $n$-qudit Pauli group $\mathcal{P}_n$ is generated by each of $X$,
$Z$, and phases in the $n$ qudits.  The qudit Clifford group is 
\[
\mathcal{C}_2 = \left\{ U \;\middle|\; U \mathcal{P}_n U^\dagger \subseteq 
\mathcal{P}_n \right\}.
\]
The hierarchy is then recursively defined in the same fashion as qubits:
\begin{align}
\mathcal{C}_1 &= \mathcal{P}_d, \\
\mathcal{C}_{k+1} &= \left\{\, U \;\middle|\; U \mathcal{P}_n U^\dagger \in \mathcal{C}_k \right\}.
\end{align}

\paragraph{Universality for $\Z_N$ qudits}
For qudits of prime dimension $d$, one defines the two-qudit
CSUM (controlled-SUM) gate by
\[
\mathrm{CSUM}\, \lvert i \rangle \lvert j \rangle
 = \lvert i \rangle \lvert i + j \bmod d \rangle.
\]

There is a qudit analogue of the Clifford+$T$ universality for qubits: 
\begin{theorem} \cite{Howard_2012,Campbell_2012, nebe2000invariantscliffordgroups, nebe2003codesinvarianttheory}
    For prime-dimensional qudits, the following gate set is universal:
\[
\text{(full single-qudit Clifford group)}
\;\;+\;\; \mathrm{CSUM}
\;\;+\;\; \text{any single non-Clifford gate}~.
\]    
\end{theorem}

\subsection{Review of transversal CCZ gate in 5+1d non-Abelian self-correcting memory}

As an example of transversal non-Clifford logical gate from automorphism symmetry, let us review the $\Z_2^3$ 2-form gauge theory construction in \cite{Hsin2024_non-Abelian} 
in spacetime dimension $D=5+1$. This 5+1d code is an example of self-correcting quantum memory with non-Abelian excitations, i.e. a non-Abelian self-correcting memory~\cite{Hsin2024_non-Abelian}.
The theory has topological action
\begin{equation}
    \omega=\pi a_1\cup a_2\cup a_3~
\end{equation}
with 2-form $\mathbb{Z}_2^3$ gauge fields $a_1,a_2,a_3$.

We focus on the swap automorphism of the gauge group that acts on the gauge fields as
\begin{equation}
    \rho:\quad (a_1,a_2)\to (a_2,a_1)~.
\end{equation}
As shown in \cite{Hsin2024_non-Abelian}, the action transforms under the automorphism by $d\alpha$ with $\alpha=\pi(a_1\cup_1 a_2)\cup a_3$, and thus the swap symmetry is
\begin{equation}
    U_\rho=(-1)^{\int (a_1\cup_1 a_2)\cup a_3}V_\rho~.
    \label{eq:CCZqubic}
\end{equation}
This swap symmetry associated with the operator \eqref{eq:CCZqubic} generates a non-Clifford gate due to the $e^{i\int\alpha}$ gauged SPT factor, with specific choices of a 5d spatial manifold $M$. 
In \cite{Hsin2024_non-Abelian}, we discussed the example of
/For instance, let us choose 
$M$ being the 5d Wu manifold, $M=\text{Wu}$. Since $H^2(\text{Wu},\Z_2)=\Z_2$ and $H^4(\text{Wu},\Z_2)=0$, the code space consists of three logical qubits.
It has been shown in Ref.~\cite{Hsin2024_non-Abelian} that due to the additional operator \eqref{eq:CCZqubic}, the above swap symmetry implements a transversal non-Clifford gate
\begin{align}
\overline{ \text{SWAP}}_{1,2} \overline{\text{CCZ}}_{1,2,3}~.
\end{align}
Remarkably, this self-correcting memory with a non-Clifford gate surpasses the distance scaling $d=\Theta(n^{2/5})$, and surpasses that of 3+1d and 6+1d color code $d=\Theta(n^{1/3})$. See Ref.~\cite{Hsin2024_non-Abelian} for detailed discussions.

\subsection{4th level Clifford hierarchy transversal gate in 2+1d qudits}

Let us consider twisted $\mathbb{Z}_3^3$ gauge theory in 2+1d with topological action
\begin{equation}
    \omega=\frac{2\pi}{3} a_1\cup a_2\cup a_3~.
\end{equation}
where $a_i$ are the three $\mathbb{Z}_3$ gauge fields. This twisted gauge theory is a non-Abelian topological order realized in a Clifford stabilizer model of $\Z_3$ qudits.
The stabilizers are
\begin{align}
    &{\cal S}^{X,1}_v=
    \left(\prod_{v\in \partial e} (X^1_e)^{s(v,e)}\right) \prod_{e',e'':\int \tilde v\cup \tilde e'\cup \tilde e''=1}\mathrm{CZ}^{2,3}_{e',e''}\;,\cr 
    &{\cal S}^{X,2}_v=
    \left(\prod_{v\in \partial e} (X^2_e)^{s(v,e)}\right) \prod_{e',e'':\int  \tilde e'\cup\tilde v\cup \tilde e''=1}\mathrm{CZ}^{1,3}_{e',e''}\;,\cr 
    &{\cal S}^{X,3}_v=
    \left(\prod_{v\in \partial e} (X^3_e)^{s(v,e)}\right) \prod_{e',e'':\int  \tilde e'\cup \tilde e''\cup \tilde v=1}\mathrm{CZ}^{1,2}_{e',e''}\;,\cr 
    &{\cal S}^i_f=\prod_{e\in\partial f} \left(Z^i_e\right)^{s(e,f)}\text{ for } i=1,2,3~,
\end{align}
where for chains $p\in\partial q$ the sign $s(p,q)=+1$ for $p$ appears in $\partial q$ with a plus sign, and $s(p,q)=-1$ for $p$ appears in $\partial q$ with a minus sign. The cochain $\tilde r$ for chain $r$ is the cochain that takes value 1 on $r$ and 0 on other chains. In the expression, 
$X^i,Z^i$ are the $\mathbb{Z}_3$ Pauli $X,Z$ operators for qutrits of species $i$, and $\text{CZ}^{i,j}_{e,e'}|r,r'\rangle=e^{\frac{2\pi i}{3} r_i r_j}|r_i,r_j\rangle$ for $i,j\in\{1,2,3\}$ is the qutrit CZ operator for the qutrit $Z^i|r_i\rangle=e^{\frac{2\pi i r_i}{3}}|r_i\rangle$ on edge $e$ of species $i$ and the qutrit $Z^j|r_j\rangle=e^{\frac{2\pi i r_j}{3}}|r_j\rangle$ on edge $e'$ of species $j$.

In fact, the theory is dual to an untwisted gauge theory with a non-Abelian gauge group given by the extension of $\mathbb{Z}_3\times\mathbb{Z}_3$ by $\mathbb{Z}_3$ specified by the cocycle $e^{\frac{2\pi i}{3}a_1\cup a_3}$. The non-Abelian gauge field in this dual description is $(a_1,a_2,\tilde a_3)$ for the dual gauge field $\tilde a_3$ of $a_3$. Here, we will focus on the description using the twisted $\mathbb{Z}_3^3$ gauge theory.

As before, we will focus on the automorphism $a_1\leftrightarrow a_2$.
The automorphism symmetry is
\begin{equation}
    U_\rho(M)= e^{\frac{2\pi i}{3}\int_M (a_1\cup_1 a_2)\cup a_3} V_\rho(M)~.
\end{equation}

Let's take the space to be $M=T^2=S^1_x\times S^1_y$. Denote the holonomy of $\mathbb{Z}_3$ gauge fields $a_1,a_2,a_3$ as $n_i^x,n_i^y\in\{0,1,2\}$ for $i=1,2,3$, the factor that decorates the automorphism symmetry gives
\begin{equation}
    e^{-\frac{2\pi i}{3}\left(n_1^xn_2^xn_3^y-n_1^y n_2^yn_3^x\right)}~.
\end{equation}
The equation of motions from the topological action imply that the holonomies need to obey the condition 
\begin{equation}
    n_i^xn_j^y-n_i^yn_j^x=0\text{ mod }3~,
\end{equation}
where we used the property that $a_1\cup_1 a_2(01)=a_1(01)a_2(01)$ on 1-cell $(01)$.
For example, $n_1^x=n_2^x=n_3^x=1$, $n_1^y=n_2^y=n_3^y=2$ is a valid set of holonomies, and the decoration factor gives $e^{\frac{4\pi i}{3}}$.

To be concrete, consider the logical subspace $n_3^x=1$. This can be implemented by measuring the logical Pauli $\bar Z=e^{2\pi i n_3^x/3}=\prod_x Z_3$ operator for the third $\mathbb{Z}_3$ along the $x$ direction using lattice surgery \cite{Horsman:2011hyt}, where $Z_3$ is the physical Pauli $Z$ on the third qutrit.

In this subspace, the other holonomies can be solved as $n_1^y=n_1^xn_3^y$, $n_2^y=n_2^xn_3^y$ from $n_1^xn_3^y=n_1^yn_3^x$, $n_2^xn_3^y=n_2^yn_3^x$. One can verify that the remaining holonomy condition $n_1^x n_2^y=n_1^yn_2^x$ is also automatically satisfied. Thus the subspace is labeled by independent holonomies $n_1^x,n_2^x,n_3^y$.
The logical gate gives
\begin{equation}
    e^{-\frac{2\pi i}{3}\left(n_1^xn_2^xn_3^y-n_1^xn_2^x(n_3^y)^2\right)}=e^{-\frac{2\pi i}{3}n_1^xn_2^x n_3^y(1-n_3^y)}~.
\end{equation}
In other words, on the logical subspace of 3 logical qutrits $\mathbb{C}[\mathbb{Z}_3^3]=\text{Span}\{|n_1^x,n_2^x,n_3^y\rangle\}$, the symmetry acts as
\begin{equation}\label{eqn:USWAPlogical}
    U_\text{SWAP}|n_1^x,n_2^x,n_3^y\rangle=e^{-\frac{2\pi i}{3}n_1^xn_2^x n_3^y(1-n_3^y)}|n_2^x,n_1^x,n_3^y\rangle~.
\end{equation}
The computation in (\ref{eqn:swappermutationdecorZN-0})-(\ref{eqn:swappermutationdecorZN}) shows that
\begin{theorem}
    The logical gate (\ref{eqn:USWAPlogical}) is realized by the finite depth circuit
\begin{equation}\label{eqn:USWAPcircuit}
    U_\mathrm{SWAP}=e^{-\frac{2\pi i}{3}\int_\text{space}(a_1\cup_1 a_2)\cup a_3 }\left(\prod\mathrm{SWAP}_{1,2}\right)=\left(\prod_{e,e':\int \tilde e\cup \tilde e'=1} \left(\mathrm{CCZ}_{e,e,e'}^{1,2,3} \right)^\dag \right)\left(\prod_{e''}\mathrm{SWAP}^{1,2}_{e'',e''}\right)~,
\end{equation}
where $\mathrm{SWAP}_{1,2}$ swaps the two physical qutrits on each edge for the first two $\mathbb{Z}_3$s. 

In the last expression, $\left(\mathrm{CCZ}_{e,e,e'}^{1,2,3} \right)^\dag|r_1,r_2,r_3\rangle=e^{-\frac{2\pi i}{3}r_1 r_2 r_3}$ for the physical qutrits $|r_1\rangle\otimes|r_2$ on the edge $e$ for the two species $1,2$, and the physical qutrit $|r_3\rangle$ on the edge $e'$ for the species $3$, respectively. 
In the above, $Z^i|r_i\rangle=e^{\frac{2\pi i r_i}{3}}|r_i\rangle$.
In the integral $\int \tilde e\cup \tilde e'$, $\tilde e$ is the 1-cochain that takes value 1 on edge $e$ and 0 on other edges, and similarly for $\tilde e'$.
 The SWAP operation $\mathrm{SWAP}_{e'',e''}^{1,2}$ exchanges the physical qutrits of species $1,2$ on the same edge $e''$.

\end{theorem}
Let us show that the logical gate (\ref{eqn:USWAPlogical}) is non-Clifford:
\begin{theorem}
    The logical gate $U_\text{SWAP}$ in (\ref{eqn:USWAPlogical}) is a non-Clifford gate of fourth Clifford hierarchy for the $\mathbb{Z}_3^3$ logical qutrits. 
\end{theorem}

\begin{proof}
    Under conjugating by Pauli $\bar X$ that shifts $n_3^y\rightarrow n_3^y+1$, the gate transforms by additional gate
\begin{equation}
e^{-\frac{2\pi i}{3}n_1^xn_2^x (n_3^y+1)(-n_3^y)-\left(-\frac{2\pi i}{3}n_1^xn_2^x n_3^y(1-n_3^y)\right)}
=
 e^{\frac{4\pi i}{3}n_1^xn_2^xn_3^y}~.
\end{equation}
Thus the gate maps Pauli $\bar X$ to $(\overline{\text{CCZ}})^2=\overline{\text{CCZ}}^\dag$ of third Clifford hierarchy. This shows that $U_\text{SWAP}$ is in the fourth Clifford hierarchy and is a non-Clifford gate.
\end{proof}

\paragraph{More general $\mathbb{Z}_N$ qudits}

We remark that instead of qutrit, we can also consider general $\mathbb{Z}_N$ qudits with $N\geq 3$. The theory is twisted $\mathbb{Z}_N^3$ gauge theory with topological action
\begin{equation}
    \omega=\frac{2\pi}{N} a_1\cup a_2\cup a_3~,
\end{equation}
for $\mathbb{Z}_N$ gauge fields $a_1,a_2,a_3$.
The only exception $N=2$ gives trivial gate since $(n_3^y)^2=n_3^y$ mod 2.
For $\mathbb{Z}_N$ qudits, the logical gate from the swap automorphism symmetry $a_1\leftrightarrow a_2$ is
\begin{equation}
    U_\text{SWAP}|n_1^x,n_2^x,n_3^y\rangle=
    e^{-\frac{2\pi i}{N}(n_1^xn_2^xn_3^y-n_1^xn_2^x (n_3^y)^2)}=|n_2^x,n_1^x,n_3^y\rangle=e^{-\frac{2\pi i}{N}n_1^x n_2^x n_3^y(1-n_3^y)}|n_2^x,n_1^x,n_3^y\rangle~.
\end{equation}
Similarly, under conjugating by Pauli $\bar X$ that changes $n_3^y\rightarrow n_3^y+1$, the gate changes by
\begin{equation}
    e^{\frac{4\pi i}{N}n_1^x n_2^x n_3^y}~.
\end{equation}
For $N\geq 3$, this is the square of the CCZ gate for $\mathbb{Z}_N$ qudits, and it belongs to the 3rd level Clifford hierarchy.
Thus the logical gates $U_\text{SWAP}$ for $N\geq 3$ give a $\mathbb{Z}_N$ qudit non-Clifford gate of 4th level $\mathbb{Z}_N$ qudit Clifford hierarchy.

\paragraph{Surpassing the generalized qubit Bravyi-K\"onig bound}

We note that the model is a $\mathbb{Z}_N$ qudit stabilizer model with $N\geq 3$, and it can do transversal logical gate in 4th Clifford hierarchy in a 2d Clifford stabilizer code. This surpasses the conjectured bound for qubits proposed in \cite{Kobayashi:2025cfh} for qubit Clifford hierarchy stabilizer models.

\subsection{New transversal logical CS gate in 3+1d $\mathbb{Z}_2\times\mathbb{Z}_2$ toric code}

Consider untwisted $\mathbb{Z}_2\times\mathbb{Z}_2$ gauge theory in 3+1d. There are two symmetries we can consider: the automorphism symmetry that swaps two $\Z_2$s, and the gauged SPT symmetry
\begin{equation}
    U=(-1)^{\int a_1\cup a_2^2}~,
\end{equation}
where $(a_1,a_2)$ are the $\mathbb{Z}_2\times\mathbb{Z}_2$ gauge fields. 

We will consider the setting where we first insert a SWAP domain wall defect in the space (which is a surface in the 3D space), and then we apply the gauged SPT symmetry at fixed time on the entire space. In such scenario, the gauged SPT symmetry is modified: along the 2D intersection of the gauged SPT symmetry and the SWAP domain wall defect, the gauged SPT symmetry $U$ is decorated with additional cochain $e^{i\int\alpha}$:
\begin{equation}
    e^{i\int \alpha}=i^{\int \left(a_1\cup a_2+da_1\cup_1 a_2\right)}~.
\end{equation}
We will show that such cochain implies that the gauged SPT symmetry gives a transversal CS logical gate.

Let us take the space to have topology of $T^3$, i.e. periodic boundary conditions in $x,y,z$ directions. The symmetries are:
\begin{itemize}
    \item SWAP domain wall defect on the $x,y$ plane at fixed $z$.
     The SWAP domain wall defect is on the $x,y$ plane at fixed $z$. The SWAP domain wall restricts the logical qubits for the two $\mathbb{Z}_2$s along the perpendicular $z$ direction to be identified. See Fig.~\ref{fig:autotwist} for an illustration. We will take the support to be the $x,y$ plane at $z=0$.
     
    \item Gauged SPT symmetry on the entire space.
    
    The total unitary operator of the gauged SPT symmetry in the presence of the automorphism defect is
\begin{align}\label{eqn:3Dautotwist}
    &U_\text{SPT}\cr 
    &=(-1)^{\int a_1\cup a_2^2}i^{\int_\text{z=0} \left(a_1\cup a_2+da_1\cup_1 a_2\right)}\cr
    &=
    \left(\prod_{e_1,e_2,e_3:\int \tilde e_1\cup \tilde e_2\cup \tilde e_3=1}\mathrm{CCZ}^{1,2,2}_{e_1,e_2,e_3}\right)
    \left(\prod_{e_4,e_5:\int_{z=0} \tilde e_4\cup \tilde e_5=1}\mathrm{CS}^{1,2}_{e_4,e_5}\right)
    \left(\prod_{f,e_6:\int_{z=0} \tilde f\cup_1 \tilde e_6=1}\prod_{e'\in\partial f}\left(\mathrm{CS}_{e',e_6}^{1,2}\right)^{s(e',f)}\right)~,\cr
\end{align}
where the first integral is over the entire space and the second integral is over the support of the automorphism defect. In the last expression, $s(e',f)=+1$ if $e'$ appears in $\partial f$ with plus sign, while $s(e',f)=-1$ if $e'$ appears in $\partial f$ with minus sign. 

\end{itemize}

\begin{figure}[t]
    \centering
    \includegraphics[width=0.4\linewidth]{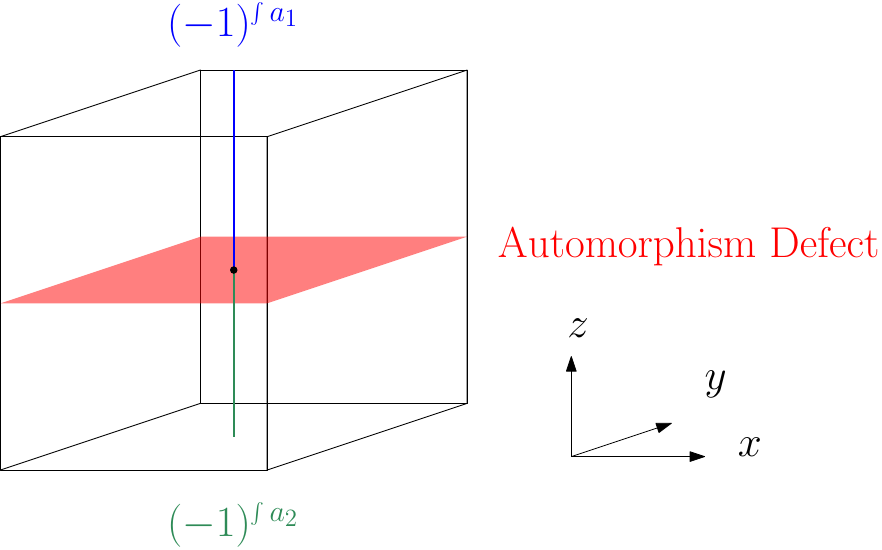}
    \caption{3D space with automorphism defect inserted at fixed $z$ coordinate. The holonomy of $a_1,a_2$ along the $z$ direction is identified.}
    \label{fig:autotwist}
\end{figure}

The action of the gauged SPT symmetry depends on whether there is a SWAP domain wall defect:
\begin{itemize}
    \item In the absence of the SWAP domain wall, the operator $U$ acts trivially on the logical states, since $a_1^2=0,a_2^2=0$ on $T^3$.

\item In the presence of the SWAP domain wall defect, applying $U$ gives logical CS gate for the qubits along the $x,y$ directions parallel to the SWAP domain wall. Denote the holonomy of $a_i$ along $x,y$ directions by $n_i^x,n_i^y$, this is
\begin{equation}
    e^{i\int \alpha}=i^{n_1^xn_2^y-n_2^xn_1^y}=\overline{CS}_{n_1^x,n_2^y}\overline{CS}_{n_1^y,n_2^x}^\dag~.
\end{equation}

\end{itemize}

\begin{theorem}
    In $\mathbb{Z}_2\times\mathbb{Z}_2$ toric code in 3D cubic lattice with periodic boundary conditions in the presence of the SWAP automorphism defect at $z=0$ plane, the following unitary circuit gives transversal logical CS gate:
\begin{equation}
    U_\text{SPT}|n_1^x,n_2^x,n_1^y,n_2^y\rangle=\overline{CS}_{n_1^x,n_2^y}\overline{CS}_{n_1^y,n_2^x}^\dag|n_1^x,n_2^x,n_1^y,n_2^y\rangle=i^{n_1^xn_2^y-n_2^xn_1^y}|n_1^x,n_2^x,n_1^y,n_2^y\rangle~.
\end{equation}
    
\end{theorem}

We remark that the transversal CS gate in the 3+1d untwisted gauge theory is in 3rd Clifford hierarchy, and this is consistent with the Bravyi-K\"onig bound for qubit Pauli stabilizer models \cite{Bravyi:2013dx}. Other non-Clifford gates such as CCZ, CS, T has been constructed using finite depth circuits in e.g. \cite{Kubica:2015mta,Barkeshli:2023bta,Fidkowski:2023dpe}.

\section{Discussion}
\label{sec:discussion}

In this work we show that automorphism symmetries in twisted gauge theories can lead to symmetry extension, higher group symmetry, and/or non-invertible symmetries. We illustrate the discussion with various twisted gauge theories in different spacetime dimensions. In addition, we show that automorphism symmetry gives rise to new non-Clifford transversal logical gates, especially transversal gates in 2+1d of 4th level Clifford hierarchy for $\mathbb{Z}_N$ qudit Clifford stabilizer models.

There are several future directions we would like to pursue. First, we would like to extend the Bravyi-K\"onig bound for qudit Clifford hierarchy stabilizer models, similar to \cite{Kobayashi:2025cfh}. This is important for understanding the natural logical gates in qudit stabilizer models. For example, some algorithm might be more efficiently implemented using qudit Stabilizer models of higher Clifford hierarchy. Second, we would like to explore the dynamics and renormalization group flows of gauge theories coupled to matter using the automorphism symmetries. For example, there can be new anomaly matching conditions for the higher symmetries induced by the automorphism symmetries. In addition, the automorphism symmetries can also constrain the relevant operators, result in potentially novel fixed points.

\section*{Acknowledgment}

We thank Guanyu Zhu for discussions and collaboration on a related work \cite{Kobayashi:2025cfh}. We thank Sakura Sch\"afer-Nameki and Alison Warman for discussion.
P.-S.H. is supported by Department of Mathematics King’s College
London. R.K. is supported by the U.S. Department of Energy through grant number DE-SC0009988 and the Sivian Fund.
R.K. and P.-S.H. thank Isaac Newton Institute for hosting the workshop ``Diving Deeper into Defects: On the Intersection of Field Theory, Quantum Matter, and Mathematics'', during which part of the work is completed.

\bibliographystyle{utphys}
\bibliography{biblio.bib,mybib_merge}
\end{document}